\documentclass{article}

\usepackage{arxiv}

\usepackage[utf8]{inputenc} 
\usepackage[T1]{fontenc}    
\usepackage{hyperref}       
\usepackage{url}            
\usepackage{amsfonts}       
\usepackage{graphicx}
\usepackage{amsthm}
\usepackage{amsmath}
\usepackage{amssymb}
\usepackage{natbib}
\usepackage{caption}
\usepackage{subcaption}
\usepackage{enumerate}
\usepackage{color}
\usepackage{multirow}
\usepackage{placeins}
\usepackage{longtable}

\usepackage{tikz}
\usetikzlibrary{automata, arrows.meta, positioning}

\usepackage[ruled]{algorithm2e}
\SetKwInput{KwInput}{Input}     
\SetKwInput{KwOutput}{Output}
\SetKwBlock{UpdateGraphBlock}{Update Graph:}{}
\usepackage{setspace}

\newtheorem{theorem}{Theorem}
\newtheorem{lemma}[theorem]{Lemma}
\newtheorem{proposition}[theorem]{Proposition}
\theoremstyle{definition}
\newtheorem{definition}[theorem]{Definition}
\newtheorem{remark}[theorem]{Remark}
\newtheorem*{example}{Example}

\newcommand{\rk}{\mathbb{R}^m}

\title{Informative
simultaneous
confidence intervals
in graphical group sequential test procedures}
\author{
 Liane Kluge\\
  Competence Center for Clinical Trials Bremen\\
  University of Bremen \\
  \texttt{liane@uni-bremen.de}\\
  \And
 Werner Brannath \\
  Institute for Statistics and\\
  Competence Center for  Clinical Trials Bremen\\
  University of Bremen \\
  \texttt{brannath@uni-bremen.de}\\
}

\begin{document}
\maketitle
\begin{abstract}
Test procedures for multiple hypotheses in a group sequential clinical trial that control the family-wise error rate are considered. Several graphical group sequential tests suggested in the literature, which are special cases of Bonferroni-closure tests, are discussed. The focus is on the question of whether to consider at the current stage only the evidence of the current repeated p-value or the evidence over all repeated p-values from the previous stages. A new test strategy controlling the family-wise error rate is introduced that consistently works across all hypotheses, with the evidence (i.e., repeated p-value) from the current stage. The strategy is more powerful than similar previously suggested test procedures. This is achieved by using the evidence from previous stages to increase the significance levels. For the test procedures, corresponding compatible simultaneous confidence intervals are presented, having the disadvantage of often not providing additional information on the treatment effects. For this reason, we extend previous work about informative simultaneous confidence intervals for one-stage graphical tests to graphical group sequential trials. Iterative algorithms are introduced that calculate these informative bounds that have a small power loss compared to the original graphical group sequential test. The boundaries can be calculated after each stage. In addition, previous work is extended by a criterion to estimate the accuracy of the numerically calculated boundaries. The suggested informative bounds can be used to provide median-conservative, i.e., reliable estimators, for estimating the treatment effects in a group sequential test with multiple hypotheses. 
\end{abstract}


\section{Introduction}
For clinical trials with multiple hypotheses and strong control of the family-wise error rate (FWER), many strategies have been suggested. Bonferroni-based closure tests offer the advantage of achieving higher power compared to simple and well-known Bonferroni-adjusted tests. \citealp{bretz2009graphical} suggested a graphical approach to perform Bonferroni-type tests based on the closure principle. Such designs have several advantages. 
The graphical tests allow for a visualization of hierarchies and preferences among the different null hypotheses. One special case is the fixed-sequence test with two hypotheses: first one hypotheses (e.g.\ for the primary endpoint) is tested at the overall significance level $\alpha$. If it can be rejected, the level $\alpha$ is shifted to the second hypothesis (e.g.\ for the secondary endpoint) so that it can be
tested at level $\alpha$. Otherwise, both hypotheses are accepted. There is also interest in testing multiple hypotheses repeatedly in time under strong control of the family-wise error rate (FWER). Group sequential designs for single null hypotheses that incorporate at least one interim analysis have already been extensively studied, see, for instance, \citealp{BuchWassmerBrannath2025}. Naively borrowing testing strategies from multiple tests for testing multiple hypotheses repeatedly in time can cause inflations of the FWER. \citealp{glimm2010hierarchical} and \citealp{tamhane2010testing} have discussed strategies for the fixed-sequence test. \citealp{maurer2013multiple} suggest for all graphical tests from \citealp{bretz2009graphical} a extension to group sequential trials. 
Two strategies are considered: One only uses the \emph{repeated p-value} of the currently conducted stage, and the other one uses the so-called \emph{sequential p-value}, which is defined as the minimum over all repeated p-values up to the conducted stage. In this paper, we will briefly discuss these strategies with regard to the question, how the data from the previous stages should be used at the current stage of data collection. In particular, this leads us to a new suggested testing strategy for multiple group sequential trials that uses for each hypothesis the final stage p-value to decide whether to reject the hypothesis, but the evidence from all repeated p-values over all stages and hypotheses is used to increase the corresponding significance level of the test. 

This paper deals with simultaneous confidence intervals for graphical group sequential tests, such as those proposed in \citealp{glimm2010hierarchical} and \citealp{maurer2013multiple} and the new suggested strategy. In addition, suitable estimators are considered. We restrict ourselves to a multiple testing problem of left-sided hypotheses $H_j:\theta_j\leq 0$, $1\leq j\leq m$, and simultaneous confidence intervals (SCI) that are bounded from the left $\text{SCI}=(L_1,\infty)\times\dots\times(L_m,\infty)$. Due to the possibility of shifting and inverting, the results can easily be extended to other one-sided null hypotheses. By intersecting confidence regions, we can also obtain SCIs for two-sided hypotheses.

In general, finding simultaneous confidence intervals for multiple testing strategies with strong control of the familywise error rate that are both compatible and informative is difficult. \emph{Compatibility} means that the confidence bounds yield the same test decisions as the given multiple test (i.e. $H_j$ is rejected if and only if $L_j\geq 0$). The term \emph{informative} refers to bounds which strictly increase with increasing evidence against the corresponding null hypotheses (unless none of the gatekeeper hypotheses were rejected), a property compatible confidence intervals often fail to have. For single-stage tests it is natural to define the evidence against a hypothesis as increasing when the corresponding (un-adjusted) p-value decreases. In this case, informative SCI bounds would almost always provide more information on $\theta_j$ (i.e.\ $L_j> 0$) if $H_j$ but not all null hypothesis are rejected. 
For multiple tests that are more powerful than the Bonferroni-adjusted test, for instance, those based on the closure principle, constructing SCI bounds that are both compatible and informative poses a problem. This also applies to the graphical group sequential tests introduced in \citealp{maurer2013multiple} since they are special cases of Bonferroni-based closure tests. One can easily use the results from \citealp{SB08} to obtain compatible simultaneous confidence bounds that will be listed in this paper. Unfortunately, the bounds remain uninformative in the sense that the borders of rejected hypotheses remain stuck at the boundary of the null hypothesis as long as not all hypotheses have been rejected. In \citealp{informativeSCIBrannathKlugeScharpenberg} we have suggested informative simultaneous confidence intervals based on single-stage graphical tests.

The aim of this paper is to extend the concept and method of informative confidence intervals suggested in \citealp{informativeSCIBrannathKlugeScharpenberg} to graphical group sequential trials, like e.g.\ in  \citealp{maurer2013multiple}. Similar to \citealp{informativeSCIBrannathKlugeScharpenberg}, we will see that the rejections made by the graphical algorithms cannot exactly be reproduced by our procedure, but our simultaneous confidence bounds will always be informative for all hypotheses (except for those with gatekeepers).
This requires a reasonable definition of \emph{evidence} for the graphical group sequential tests, which will also be provided in this paper. The informativeness yields that the borders do not stack at the borders of the hypotheses, except for the negligible case that the repeated or sequential p-value is equal to the final local level. To the best of our knowledge, this is not possible for the original graphical group sequential tests.

As for the single-stage case, our bounds depend on an information weight $q\in(0,1)$. The weight indicates the trade-off between informativeness and compatibility with the original group sequential graph. If the information weight tends to zero, the test decisions given through the information bounds become increasingly similar to those of the underlying group sequential graph. For greater information weights, one could roughly say that the expected size of the confidence bounds and, therefore, the informativeness increase, at the expense of the rejections made. The rejections increasingly differ from those of the original group sequential graph.

Although there is no compatibility with the original graph, our simultaneous confidence bounds are, of course, consistent with the test implied by them. 
Moreover, they all have the monitoring property, i.e.\ they can be calculated at every stage, irrespective of any stopping rule. 

For the calculation of the informative bounds, we will give iterative algorithms that all yield (conservative) lower and (anti-conservative) upper approximations of the simultaneous confidence bounds. The algorithms are based on the algorithm from \citealp{informativeSCIBrannathKlugeScharpenberg}, which is supplemented by also calculating an upper approximation. The calculation of an upper approximation is new and will be mathematically justified. The upper approximation offers the chance to estimate the accuracy of the numerical approximation of the lower bounds. This is of practical interest. 

As we will see, all the presented strategies for obtaining (compatible or) informative SCIs can also be used to construct common median-conservative estimators for the treatment effects. The need for unbiased estimations in adaptive (and group sequential) trials is an important topic as mentioned in (a draft of the) \citealp{ICH-E20} guideline.

In Section~\ref{MainChapterGSDGraphical}, we describe the underlying setup of this paper and the new suggested testing strategy. This includes the graphical tests introduced in \citealp{maurer2013multiple} for group sequential trials. Further, we will give a formal definition of (informative) simultaneous confidence intervals and appropriate definitions of increasing evidence. Section~\ref{SubsectionCompatibleBounds} deals with compatible simultaneous confidence bounds for the graphical group sequential designs. In Section~\ref{mainChapterInformativeSCIs}, the concept of the construction of informative SCI bounds is introduced, and iterative algorithms for the calculation of informative bounds are listed. Section~\ref{SectionMedianConservatEstimators} emphasizes that meaningful median conservative estimators of the treatment effects can be defined using our algorithms for calculating informative bounds. 
We conclude with a summary and discussion in Section~\ref{sectionSummaryAndDiscussion}.
\section{Graphical tests and informativeness for GSD}\label{MainChapterGSDGraphical}
\subsection{Set-up}\label{subsectionGraphicalGSDSetUp}
Graphical tests (for $m\in\mathbb{N}$ hypotheses) are more powerful than the (weighted) Bonferroni test. They are based on the closure principle using weighted Bonferroni tests for the intersection hypotheses and provide strong control of the FWER for a $\alpha\in(0,1)$ (see \citealp{bretz2009graphical}). For performing such a graphical test, one needs to specify $m$ hypotheses $H_j$, $j\in\{1,\dots,m\}$. We will work in this paper with the left-sided hypotheses $H_j=H_j^0:\theta_j\leq 0$. Additionally,  initial local significance levels $\alpha_j=\alpha_j(I)=\omega_j(I)\cdot\alpha\geq 0$, $I=\{1,\dots,m\}$, such that $\sum_{j=1}^m\alpha_j=\alpha$ need to be defined. They can be visualized with nodes with the weight $\alpha_j$. In addition, a transition matrix $(g_{ij})_{1\leq i,j\leq m}$ with $g_{i,j}\in[0,1]$ where each row sum is at most $1$ needs to be specified. The transition weights can be visualized by arrows between the nodes, indicating the percentage of the (initial) level that is shifted after rejecting a hypothesis to the other ones. One example is the hierarchical test, for instance, with four hypotheses (see Figure~\ref{PictureHTest}). The hierarchical test is characterized by $\alpha_1(I)=\alpha_1=\alpha$, $\alpha_j(I)=\alpha_j=0$, $j=2,3,4$, $I=\{1,2,3,4\}$ and $g_{1,2}=g_{2,3}=g_{3,4}=1$ and $g_{\ell,i}=0$ for all other transition weights.
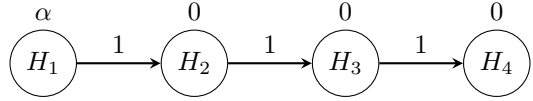
\begin{figure}[h]
\begin{center}
\begin{tikzpicture} [node distance = 2cm, on grid, auto]

\node[label={$\alpha$}] (H1) [state] {$H_1$};
\node[label={$0$}] (H2) [state, right = of H1] {$H_2$};
\node[label={$0$}] (H3) [state, right = of H2] {$H_3$};
\node[label={$0$}] (H4) [state, right = of H3] {$H_4$};

\path [-stealth, thick]
	(H1) edge  node {$1$}    (H2)
	(H2) edge  node {$1$}    (H3)
    (H3) edge  node {$1$}    (H4);
	
\end{tikzpicture}
\end{center}\caption{Hierarchical test with 4 hypotheses.}\label{PictureHTest}
\end{figure}

The levels and transition weights are updated after each rejection step by a pre-specified fixed rule. The corresponding current level at which a non-rejected hypothesis is tested is given by $\alpha_j(J)=\omega_j(J)\cdot\alpha$ where $J\subseteq I$ is the index set of not yet rejected hypotheses. The update rules can be found in \citealp{bretz2009graphical} and are also given by the ``Update Graph'' block in Algorithm~\ref{AlgorithmusMauerBretz}. The weights of the Bonferroni tests fulfil a monotonicity condition (namely $\omega_j(J)\leq\omega_j(J^{\prime})$ for $J^{\prime}\subseteq J\subseteq I$, $j\in I$) that implies \emph{consonance} of the closure test (see \citealp{hommel2007powerful}). This implies that the graphical test needs at most $m$ steps. The graphical group sequential tests from \citealp{maurer2013multiple} maintain the consonance property. The tests use group sequential weighted Bonferroni tests for the intersection tests, where for each hypothesis and stage, the \emph{nominal} significance levels are increasing in the overall significance level. 
    
More precisely, we assume in the following: We have an underlying graphical test as described above. Additionally, the conducted trial for testing the (multiple) hypotheses $H_j$, $j\in\{1,\dots,m\}$, incorporates $K\in\mathbb{N}$ analysis time points $1,\dots,K-1,K$ ($K-1$ interim looks and the final analysis $K$). For each hypothesis $H_j$, $j\in\{1,\dots,m\}$, it is assumed that the data collected up to time point $k$ can be summarized by the (asymptotic) normal test statistic $Z_{j,k}\sim\mathcal{N}(\theta_j\sqrt{I_{j,k}}\, ,1)$ with corresponding p-value $p_{j,k}=p_{j,k}(0)=1-\Phi(Z_{j,k})=1-\Phi(\hat{\theta}_{j,k}/\mathrm{SE}_{j,k})$ for testing $H_j=H_j^0$ where $\hat{\theta}_{j,k}/\mathrm{SE}_{j,k}=Z_{j,k}$. The relative amount of statistical information available at time point $k$, $1\leq k\leq K$, for hypothesis $H_j$, $j\in\{1,\dots,m\}$, is denoted by $t_{j,k}=I_{j,k}/I_{j,K}\in[0,1]$ (\emph{normalized look time}). In group sequential trials, the covariances of the test statistics $Z_{j,k}$ at different points $1\leq k\leq K$ in time are utilised. Typically, $\mathrm{cov}(Z_{j,k}, Z_{j,k^{\prime}})=\sqrt{I_k/I_{k^{\prime}}}$ for $1\leq k\leq k^{\prime}\leq K$. For a group sequential test at significance level $\gamma\in(0,1]$, for each interim analysis and the final analysis, a spent level $\alpha_{j,k}(\gamma)$ must be specified, summing up to $\gamma$ across the stages. Because it is usually difficult or even impossible to achieve the pre-planned sample sizes, it makes sense to specify such spent levels via a spending function $a(\gamma,t_{j,k})$, $t_{j,k}\in[0,1]$ that is increasing in $t_{j,k}$ with $a(\gamma,0)=0$ and $a(\gamma,1)=\gamma$. Now, the decision boundaries or the \emph{nominal levels} $\alpha_{j,k}^{*}(\gamma)$, $1\leq k\leq K$, are defined step by step through the stages via
\begin{align*}
\mathbb{P}_0\left(\{p_{j,k}(0)\leq \alpha_{j,k}^*(\gamma)\}\cap\bigcap\nolimits_{s=1}^{k-1}\{p_s >\alpha_{j,k}^{*}(\gamma)\}\right)=a(\gamma, t_{j,k})-a(\gamma, t_{j,k-1})=\alpha_{j,k}(\gamma)~.
\end{align*}
Throughout this paper, we assume no binding stop for futility at each stage. 

A hypothesis $H_j$ can be rejected at stage $k$ of a group sequential test under control of the level $\gamma$ if for the local p-value it holds $p_{j,k}(0)\leq\alpha_{j,k}^*(\gamma)$. If for all $\gamma\in(0,1]$ such nominal levels are defined and they are \emph{increasing} and \emph{continuous} in $\gamma$, one can also work with \emph{adjusted p-Values} at stage $k$ like the \emph{repeated p-values} defined by $p_{j,k}^r(0)=\sup\{\gamma:p_{j,k}(0) >\alpha_{j,k}^*(\gamma)\}$ (see \citealp{posch2008note}). Then, a rejection at stage $k$ is equivalent to $p_{j,k}^r(0)\leq\gamma$. If one is interested in a rejection at any time point up to $k$, $1\leq k\leq K$, one can work with the adjusted \emph{sequential p-value} $p_{j,k}^s(0)=\min\{1\leq s\leq k:p_{j,k}^r(0)\}$, $j\in\{1,\dots,m\}$ (see \citealp{liu2008adaptive}). Then, $p_{j,s}^r(0)\leq\alpha_{j,s}^*(\gamma)$ for any $s\leq k$, is equivalent to $p_{j,k}^s(0)\leq\gamma$. 

For the graphical group sequential tests from \citealp{maurer2013multiple}, it is crucial that the nominal levels $\alpha_{j,k}^*(\gamma)$ are increasing and continuous in $\gamma\in(0,1]$. For the original one-stage graph it holds $\omega_j(J)\cdot\alpha\leq\omega_j(J^{\prime})\cdot\alpha$ for $j\in J^{\prime}\subseteq J\subseteq I$ and with the monotonicity property of the nominal levels, it follows $\alpha_{j,k}^{*}(\omega_j(J)\cdot\alpha)\leq \alpha_{j,k}^{*}(\omega_j(J^{\prime})\cdot\alpha)$, $1\leq k\leq K$. When working with the repeated or sequential p-values (in the case of monotonicity and continuity), it is also immediately apparent that the monotonicity is maintained. In this paper, we even work with the assumption that the nominal levels are strictly increasing (and continuous) in $\gamma\in (0,1]$. This is required to guarantee continuity properties and monotonicity of the repeated or sequential p-values when the observed sample remains the same, and one considers the p-values for the shifted null hypotheses $H_j^{\mu_j}:\theta_j\leq\mu_j$, $\mu_j\in\mathbb{R}$, $j\in\{1,\dots,m\}$. In \citealp{maurer2013multiple} it was shown that for $\gamma^{\prime}\geq\gamma$ it holds: If for all $s=1,\dots,k$ the spent levels fulfil $\alpha_{j,s}^*(\gamma^{\prime})\geq \alpha_{j,s}^*(\gamma)$ this implies $\alpha_{j,k}^*(\gamma^{\prime})\geq\alpha_{j,k}^*(\gamma)$ of the nominal levels, for a $1\leq k\leq K$, $j\in\{1,\dots,m\}$, $\gamma,\gamma^{\prime}\in (0,1]$. It can easily be seen from the corresponding proof in \cite{maurer2013multiple} that if the spent levels $\alpha_{j,s}(\gamma)$ are \emph{strictly} increasing in $\gamma$ across all analyses $s=1,\dots,k$ this also holds for the nominal significance levels $\alpha_{j,s}^{*}(\gamma)$. As mentioned in \citealp{maurer2013multiple}, for the \emph{Pocock} boundaries mimicked by the spending function $a(\gamma,t)=\gamma\ln(1+(e-1)t)$ (see \citealp{lan1983discrete}) and the so-called power spending functions $a(\gamma, t)=\gamma\cdot t^{\rho}$, $\rho>0$ (see \citealp{kim1987design}) the spent levels and thus the nominal levels are indeed strictly increasing in $\gamma$. For the O'Brien-Fleming (OBF) boundaries mimicked by $a(\gamma, t)=2(1-\Phi(\Phi^{-1}(1-\gamma/2)/\sqrt{t}))$ (see \citealp{lan1983discrete}) the monotonicity could be proven on $(0,0.318)$. 

In this paper, we assume like in \citealp{informativeSCIBrannathKlugeScharpenberg} that the standard errors $\mathrm{SE}_{j,k}$ introduced above are independent from the true $\theta\in\mathbb{R}^m$ and consider for each $\mu_j\in\mathbb{R}$ the shifted local p-value $p_{j,k}(\mu_j)=1-\Phi((\hat{\theta}_{j,k}-\mu_j)/\mathrm{SE}_{j,k})$, for the shifted hypotheses $H_j^{\mu_j}:\theta_j\leq\mu_j$, $j\in\{1,\dots,m\}$, $1\leq k\leq K$. This allows us to define the family of \emph{repeated} and \emph{sequential} p-values 
\begin{align}
    p_{j,k}^r(\mu_j)=\sup\{\gamma: p_{j,k}(\mu_j)>\alpha_{j,k}^{*}(\gamma)\}\quad\text{and}\quad p_{j,k}^s(\mu_j)=\min\{p_{j,s}^r(\mu_j):1\leq s\leq k\}~,\quad \mu_j\in\mathbb{R}, \label{DefinitionRepeatedSequentialPValues}
\end{align}
for $1\leq j\leq m$ and $1\leq k\leq K$. If the nominal levels are strictly increasing and continuous in $\gamma\in(0,1]$ as discussed above, this yields the following properties of the repeated and sequential p-values: Let $\lambda\in\{r,s\}$. The p-values $\mu_j\mapsto p_{j,k}^{\lambda}(\mu_j)$ are continuous and strictly increasing on $(-\infty, \min\{\mu_j:p_{j,k}(\mu_j)\geq\alpha_{j,k}^*(1)\}]$ with $p_{j,k}^{\lambda}(\mu_j)\overset{\mu_j\rightarrow\infty}{\longrightarrow}1$ and $p_{j,k}^{\lambda}(\mu_j)\overset{\mu_j\rightarrow -\infty}{\longrightarrow}0$. 

\subsection{Group sequential trials using graphical approaches}\label{subsectionGraphicalGSD}
The following algorithm combines the two methods from \citealp{maurer2013multiple} and generalizes it to the case where individual hypotheses are not considered beyond some stage even if not rejected, i.e. is stopped for futility. 

\begin{algorithm}[h]
\setstretch{1.1}
\caption{Graphical test for group sequential designs with $\lambda\in\{r,s\}$ from \citealp{maurer2013multiple}}\label{AlgorithmusMauerBretz} 
\KwInput{Group sequential graphical test with $K$ stages for $m$ hypotheses, overall significance level $\alpha$ and nominal significance levels $\alpha_{j,k}^*(\gamma)$, $\gamma\in[0,1]$, $1\leq k\leq K$, $1\leq j\leq m$.}
Set $I=D=\{1,\dots,m\}$, $\bar R_0^{\lambda}=\varnothing$, $k=1$, and for all $1\leq j\leq m$ set $k_j^*=1$\;
\While{$|I|\geq 1$, $|D|\geq 1$ \emph{\textbf{and}} $k\leq K$}{
Continue the trial with the hypotheses $H_j$ for $j\in D$ and compute local p-values $p_{j,k}$ and group sequential p-values $p_{j,k}^{\lambda}$\;
Initialize $\bar R_k^{\lambda}=\bar R_{k-1}^{\lambda}$\;
\While{$\{j\in I: p_{j,k_j^*}^{\lambda}\leq\alpha_j(I)\}\neq\varnothing$}{
Choose a $j\in \{j\in I: p_{j,k_j^*}^{\lambda}\leq\alpha_j(I)\}$\;
\UpdateGraphBlock{
$I\rightarrow I\setminus\{j\}$\;
$\alpha_\ell(I)\rightarrow \begin{cases}
			\alpha_{\ell}(I)+\alpha_j(I)g_{j\ell} & \ell\in I\\
            0 & \text{otherwise}
		 \end{cases}$\;
$g_{\ell i}\rightarrow\begin{cases}
			\frac{g_{\ell i}+g_{\ell j}g_{j i}}{1-g_{\ell j}g_{j\ell}} & \ell,i\in I,\ell\neq i, g_{\ell j}g_{j\ell}<1\\
            0 & \text{otherwise}
		 \end{cases}$\;
}
Reject $H_j$ and update index set of rejected hypotheses up to stage $k$ by $\bar R_k^{\lambda}\rightarrow \bar R_k^{\lambda}\cup\{j\}$\;
}
Decide for each $j\in D$ whether to stop data collection and update in this case $D\rightarrow D\setminus\{j\}$ and set $\tau_j^*=k$\;
Update $k\rightarrow k+1$\;
Update for all $j\in D$ the current stage $k_j^*\rightarrow k$\;
}
Let $\tau^*=\max_{1\leq j\leq m}\{\tau^*_{j}\}=k$ be the stage where the trial is terminated\;
\Return{Index set $\bar R^{\lambda}:=\bar R^{\lambda}_{\tau^*}$ of rejected hypotheses.}
\end{algorithm}
In Algorithm~\ref{AlgorithmusMauerBretz}, we have denoted the index sets of the rejected hypotheses with an overline to distinguish them from the index sets of hypotheses rejected by our informative procedures that will be introduced later. Let us further explain the role of the index set $D$ and $I$. The index set $I$ only refers to the hypotheses that are not rejected so far, and the index set $D$ contains all hypotheses for which the data collection is continued. There could be the case that the data collection for a hypothesis is stopped due to futility reasons, then the corresponding index is no longer contained in $D$ but in $I$. From a practical perspective, it may also be sensible to continue collecting data on some endpoints despite rejection. In this case, $j\notin I$ but $j\in D$. Additionally, it could be attractive to continue the data collection for a rejected hypothesis $H_j$ to obtain more informative simultaneous confidence bounds. For all hypotheses, we denoted the current last stage of data collection $k_j^*=k_j^*(k)$ and the last stage of data collection $\tau_j^*$, $1\leq j\leq m$, as well as the stage $\tau^*$ where the whole trial is terminated with a star to emphasize the dependency on $D$, i.e. the decisions, and the data dependency. 

Algorithm~\ref{AlgorithmusMauerBretz} controls for both $\lambda\in\{r,s\}$ the FWER in the strong sense because the test procedure is always more conservative than performing the graphical test using for all hypotheses the sequential p-value $p_{j,K}^s$, $1\leq j\leq m$. Please note that the graphical test keeps the FWER because it is based on a uniquely defined closure test as shown in \citealp{bretz2009graphical}. In particular, the significance levels $\alpha_{\ell}(I)$, $\ell\in I$, are unique.

We will now illustrate Algorithm~\ref{AlgorithmusMauerBretz} using the example of a fixed-sequence test.

\begin{example}
    We assume a fixed sequence test for $\alpha=0.025$ with four hypotheses $H_j$, $j\in I=\{1,\dots,4\}$, and two stages ($K=2$), where a new treatment is tested in four different doses against a control treatment. Remember, that the hierarchical test is characterized by $\alpha_1(I)=\alpha_1=\alpha$, $\alpha_j(I)=\alpha_j=0$, $j=2,3,4$, $g_{1,2}=g_{2,3}=g_{3,4}=1$ and $g_{\ell,i}=0$ for all other transition weights. First, we consider the case that Algorithm~\ref{AlgorithmusMauerBretz} is executed for $\lambda=r$. The first stage ($k=1$) is conducted for all hypotheses, and the corresponding repeated p-values are calculated. The results are shown in Table \ref{tab:hiertest4Hyp}. At interim, $H_1$ is tested at level $\alpha$ and can be rejected due to $p_{1,1}^r=0.02<0.025=\alpha=\alpha(I)$. Now the graph is updated, and the level $\alpha$ is shifted to $H_2$. This hypothesis cannot be rejected at the first stage because $p_{2,1}^r=0.04>0.025=\alpha(\{2,3,4\})$. This implies that at interim, no level is shifted to $H_3$ and $H_4$, which are also not rejected. The indices of rejected hypotheses after the first stage are $\bar R_{1}^r={1}$. Let the data collection for $H_1$ be stopped now. Thus, the final conducted stage for $H_1$ equals $\tau^*_1:=1$. For all other hypotheses the study is continued to the second, final stage. The second hypothesis $H_2$ can be rejected after the second stage because $p_{2,2}^r=0.02<0.025=\alpha(\{2,3,4\})$. The graph is again updated, and the level $\alpha$ is shifted to $H_3$. The hypothesis is not rejected because $p_{3,2}^r=0.03>\alpha$. Since the final stage has already been reached, the test procedure is stopped, and $H_3$ and $H_4$ remain accepted. It follows that $\tau^*=2$ is the stage where the trial is terminated. The set of the rejected hypotheses is given by $\bar R^r=\bar R_2^r=\{1,2\}$.

    We assume now that Algorithm~\ref{AlgorithmusMauerBretz} is executed for $\lambda=s$. Because for the first stage the sequential p-values match the repeated ones (see also Table \ref{tab:hiertest4Hyp}), we have the same results as for $\lambda=r$. Only $H_1$ can be rejected after the first stage, i.e., $\bar R_1^s=\{1\}$. Again let the data collection for $H_1$ be stopped at interim and for all other hypotheses the second stage is conducted. Because a sequential p-value is defined as the minimum over the repeated p-values of stage $1$ and $2$, the hypothesis $H_2$ can also be rejected because $p_{2,2}^s=0.02<0.025$ and the level $\alpha$ is shifted to $H_3$. Compared to the $\lambda=r$ approach, it is permissible to ``look back'' and reject the hypothesis retrospectively via the first stage. We have $p_{3,2}^s=\min\{p_{3,1}^r,p_{3,2}^r\}=p_{3,2}^r=0.02<\alpha$ and can reject $H_3$. The graph is updated and the level $\alpha$ is shifted to the last hypothesis, which can also be rejected (via the first and second stage). In summary we can reject all hypotheses and receive $\bar R^s=\bar R_2^s=\{1,\dots,4\}$.

\begin{table}[ht]
\centering
\begin{tabular}[t]{c|cc| cc}
$j$ & $p_{j,1}^r$ & $p_{j,2}^r$ & $p_{j,1}^s$ & $p_{j,2}^s$\\
\hline
$1$ & $\mathbf{0.02}$ & $0.03$ & $\mathbf{0.02}$ & $0.03$\\
$2$ & $0.04$ & $\mathbf{0.02}$ & $0.04$ & $\mathbf{0.02}$ \\
$3$ & $0.02$ & $0.03$ & $0.02$ & $\mathbf{0.02}$\\
$4$ & $0.02$ & $0.01$ & $0.02$ & $\mathbf{0.01}$\\
\end{tabular}
\caption{\vspace{1em} Repeated and sequential p-values of fixed-sequence test with $K=2$ stages and level $\alpha=0.025$.}
\label{tab:hiertest4Hyp}
\end{table}%
\end{example}

\citealp{glimm2010hierarchical} and \citealp{maurer2013multiple} prefer the repeated variant of Algorithm~\ref{AlgorithmusMauerBretz} over the sequential variant. The reason is that ``looking back'' and rejecting a hypothesis with a significant result from a previous stage, even if it is no longer significant at the current stage, would not be convincing, and the power gain when allowing to ``look back'', i.e., using the $\lambda=s$ variant, would be very small.

Algorithm~\ref{AlgorithmusMauerBretz} incorporates (for both variants) the convention to consider hypotheses that are rejected at an interim analysis as rejected at all subsequent time points. If the data collection of a hypothesis is still continued despite rejection, the $\lambda=r$ variant is backward-inconsistent: For hypotheses that have already been rejected, one is satisfied with the evidence from the earlier stages, but for hypotheses that have not yet been rejected, the common data up to the current stage must be used. As an illustration, we assume for our example above that after the first stage and the rejection of $H_1$, the second stage is now conducted for all hypotheses, including $H_1$. Assume that a non-significant result $p_{1,2}^r=0.03$ is observed. This is ignored by Algorithm~\ref{AlgorithmusMauerBretz}, and $H_1$ remains rejected. In contrast to this, the third hypothesis $H_3$ could also be rejected retrospectively via the first stage because of $p_{3,1}^r=0.02<0.025$, but this is not accepted because data from the second stage is already available. Due to this backward inconsistency for $\lambda=r$ that arises when continuing to collect data of rejected hypotheses, in the following, we will only consider the case where $D\subseteq I$ in Algorithm \ref{AlgorithmusMauerBretz}.

In comparison to this, the $\lambda=s$ variant is backward-consistent, but of course, it does not work with the common evidence up to the current stage. In the following, we would like to discuss methods that work consistently with the common evidence up to the current stage and that are as powerful as possible. A starting point is our example, where using the common evidence up to the current stage implies for $H_1$ a retest at the second stage that leads to an acceptance of $H_1$. For the other hypotheses, of course, using the evidence up to the current stage means working with the repeated p-values $p_{j,2}^r$, $j=2,3,4$. However, what is not clear at first glance is which significance level should be used for retesting. \citealp{glimm2010hierarchical} suggested proceeding \emph{stagewise hierarchical}. For our example, this means that no significance level can be transferred from $H_1$ to $H_2$ and therefore also not to the other hypotheses. This implies that none of the hypotheses can be rejected. This approach can, of course, be used for all graphical tests and then looks simply as follows: at each stage $k$, $1\leq k\leq K$, the graphical test is restarted using for all hypotheses the repeated p-values corresponding to the actual last stage of data collection $k_j^*=k_j^*(k)$, $j=1,\dots,m$. This strategy uses the cumulative evidence by definition, but is conservative. 

Below, we suggest a more powerful testing strategy formalized by Algorithm~\ref{alg_compromise_compatible}. On the one hand, it always uses the data and thus the evidence up to the current stage to test a hypothesis. On the other hand, however, possible rejections at earlier stages are used to test the other hypotheses at a higher level. The strategy can be understood as a \emph{compromise} between the conservative testing strategy of performing the graphical test solely at the current stage and the liberal sequential graphical test ($\lambda=s$).

\begin{algorithm}[h]
\setstretch{1.2}
\caption{Graphical test for group sequential designs with efficient multiple adjustment} \label{alg_compromise_compatible}
\KwInput{Group sequential graphical test with $K$ stages for $m$ hypotheses, index set $\bar{R}^s$ of hypotheses rejected by the $\lambda=s$ variant of Algorithm~\ref{AlgorithmusMauerBretz} with corresponding last stages of data collection $\tau^*_j$  for each hypothesis $H_j$, p-values $p_{j,\tau^*_j}^{r}$, $1\leq j\leq m$, significance level $\alpha$.}
Set $\bar R^c=\varnothing$ and $I=\{1,\dots,m\}$\;
\For{$j\in\bar R^s$}{
\If{$p_{j,\tau^*_j}^r(0)\leq\alpha_j(\{j\}\cup I\setminus \bar R^s)$}{
\vspace{0.3em}
Reject $H_j$\;
$\bar R^c\rightarrow\bar R^c\cup\{j\}$\;
}
}
Retain all hypotheses $H_j$, $j\in I\setminus\bar R^c$\;
\Return{Index set $\bar R^c$ of rejected hypotheses.}
\end{algorithm}
In order to test all hypotheses at the highest possible level, Algorithm~\ref{AlgorithmusMauerBretz} is performed for $\lambda=s$, which provides a set of rejected hypotheses $\bar{R}^s$. Each of the hypotheses  $H_j$, $j\in\bar{R}^s$, is reconsidered individually and retested at the significance level obtained by the graphical test when keeping all other rejections except for $H_j$. Note that all hypotheses $H_j$ with $j\in I\setminus \bar{R}^s$ remain accepted. If a hypothesis $H_j$ cannot be rejected with the p-value $p_{j,\tau^*_j}^s$ via Algorithm~\ref{AlgorithmusMauerBretz} ($\lambda=s$), a rejection is also not possible with the more conservative p-value $p_{j,\tau^*_j}^r\geq p_{j,\tau^*_j}^s$. The procedure keeps the FWER because this holds for the sequential variant of Algorithm~\ref{AlgorithmusMauerBretz}. For reasons of simplification, Algorithm \ref{alg_compromise_compatible} is formulated as a ``follow-up algorithm'' applied after Algorithm \ref{AlgorithmusMauerBretz} for $\lambda=s$. Of course, one could also execute our follow-up algorithm after each stage $1\leq k\leq\tau^*$ of Algorithm  \ref{AlgorithmusMauerBretz} with $\lambda=s$.

We apply Algorithm~\ref{alg_compromise_compatible} to our example of a hierarchical test with four hypotheses. As argued above, we obtain $\bar{R}^s=\{1,\dots,4\}$. Keeping all rejections except the one under consideration leads to retesting each hypothesis at level $\alpha$ with the final repeated p-value $p_{j,\tau^*_j}^r$, $j=1,\dots, 4$. Using the results from Table \ref{tab:hiertest4Hyp}, we obtain significant results for $H_2$ and $H_4$, i.e. $\bar{R}^c=\{2,4\}$.

Note that one could start in Algorithm~\ref{alg_compromise_compatible} also with the repeated variant of Algorithm~\ref{AlgorithmusMauerBretz} and retest the hypotheses from $\bar{R}^r$. However, this strategy is more conservative and therefore not recommended. For our example, the first two hypotheses would be retested at level $\alpha$, and only $H_2$ could be rejected.

\subsection{Informativeness}
For our multiple test problem with the one sided null hypotheses $(H_j:\theta_j\leq 0)_{1\leq j\leq m}$ let us consider left-bounded simultaneous confidence intervals $\mathrm{SCI}=(L_1,\infty)\times\dots\times(L_m,\infty)$. For the graphical group sequential test introduced in Section~\ref{subsectionGraphicalGSD}, we are interested in SCIs that are ideally compatible and informative. \emph{Compatibility} means that  $H_j$ is rejected if and only if $L_j\geq 0$, $j\in\{1,\dots,m\}$. The understanding of \emph{informativeness} is based on the concept from \citealp{informativeSCIBrannathKlugeScharpenberg}. Roughly summarised, informativeness means that the bounds $L_j$  strictly increase with increasing evidence against the corresponding null hypothesis $H_j$, $1\leq j\leq m$ (part (b) in the following Definition~\ref{definitionInformativität}) and whenever possible, information about the unknown parameter $\theta_j\in\mathbb{R}$ is provided (part (a) in the following Definition~\ref{definitionInformativität}). This is formalised by:
\begin{definition}\label{definitionInformativität}
        We call a SCI with lower bounds $L=(L_1,\dots,L_m)$ \emph{informative} if
        \begin{itemize}
            \item[(a)] $L_j>-\infty$ whenever $H_j$ has no gatekeeper or for at least one gatekeeper $H_i$ for $H_j$ we have $L_i>0$;
            \item[(b)] $L_j(x^{\prime})>L_j(x)$, if the following holds for two data sets $x^{\prime}$ and $x$:
            \begin{itemize}
                \item[(i)] $x^{\prime}$ provides more \emph{evidence} against $H_j$ than $x$, and
                \item[(ii)]for all $i\neq j$ the evidence in $H_i$ is stronger in $x^{\prime}$ or the same in both data sets, and
                \item[(iii)]$L_j(x)>-\infty$.
            \end{itemize}
        \end{itemize}
    \end{definition}
The term of increasing and constant \emph{evidence} for Definition~\ref{definitionInformativität} in the context of graphical group sequential tests remains to be specified. For single-stage tests \citealp{informativeSCIBrannathKlugeScharpenberg} defined (strictly) increasing evidence via (strictly) decreasing local p-values $p_j(\mu_j)$ for all shifted null hypotheses $H_j^{\mu_j}:\theta_j\leq\mu_j$, $\mu_j\in\mathbb{R}$, $1\leq j\leq m$. For group sequential tests there is no clear understanding how to order two sample points $x,x^{\prime}$ if the data collection is stopped at different stages. This is why in the following we only compare samples for which the (current) last stage of data collection is the same. In addition, it is natural that the understanding of evidence is influenced by which local p-values of which stages may be used at the current stage for rejection of the hypotheses, i.e. whether repeated or sequential p-values may be used. This leads to the following definition.
   \begin{definition}\label{DefinitionEvidenzBretzAlgo}
        Let $\lambda\in\{r,s\}$. We say the data set $x^{\prime}$ provides \emph{more evidence} than $x$ against $H_j$, $j\in\{1,\dots,m\}$, (at stage $k$, $1\leq k\leq K$) if 
        \begin{itemize}
            \item[(a)] $k^*_j(x^{\prime})=k^*_j(k,x^{\prime})=k^*_j(k,x)=k^*_j(x)$, i.e., the current (at time point $k$) last stage of data collection for $H_j$ is the same, and 
            \item[(b)] $p_{j,k^*_j(x^{\prime})}^{\lambda}(x^{\prime},\mu_j)<p_{j,k^*_j(x)}^{\lambda}(x,\mu_j)$ for all $\mu_j\in\mathbb{R}$ where $p_{j,k^*_j(x^{\prime})}^{\lambda}(x^{\prime},\mu_j)<1$.
        \end{itemize}
        If (a) holds, and the inequality in (b) is not always strict, we say the \emph{evidence in $x$ and $x^{\prime}$ is the same}.
    \end{definition}
If $\lambda=r$ is chosen one can immediately see from the definition of the repeated p-value family from (\ref{DefinitionRepeatedSequentialPValues}) that strictly decreasing local p-values at stage $k_j^{*}$, i.e. $p_{j,k_j^*}(x^{\prime},\mu_j)<p_{j,k_j^*}(x,\mu_j)$ for all $\mu_j\in\mathbb{R}$ and strictly increasing and continuous nominal levels imply that (b) from Definition~\ref{DefinitionEvidenzBretzAlgo} is met. For the choice of $\lambda=s$, strictly decreasing local p-values at the current last stage of data collection or a fixed previous stage does not necessarily imply property (b) of Definition~\ref{DefinitionEvidenzBretzAlgo} in contrast to $\lambda=r$. The reason is that the sequential p-value is for all $\mu_j$, coined by the lowest repeated p-value (over the stages), which can differ for different $\mu_j,\mu_j^{\prime}\in\mathbb{R}.$

For multiple adjusted test strategies, we use the understanding of evidence given by the following definition. 
\begin{definition}\label{DefinitionEvidenzMixture}
        In case of following the efficient multiple adjustment-strategy, we say the data set $x^{\prime}$ provides \emph{more evidence} than $x$ against $H_j$, $j\in\{1,\dots,m\}$, if 
        \begin{itemize}
            \item[(a)] $\tau^*_j(x^{\prime})=\tau^*_j(x)$, i.e., the stage where the data collection for $H_j$ stops is the same, and 
            \item[(b)] $p_{j,\tau^*_j(x^{\prime})}^r(x^{\prime},\mu_j)< p_{j,\tau^*_j(x)}^r(x,\mu_j)$ for all $\mu_j\in\mathbb{R}$ where $p_{j,\tau^*_j(x^{\prime})}^r(x^{\prime},\mu_j)<1$, and
            \item[(c)] $p_{j,\tau^*_j(x^{\prime})}^s(x^{\prime},\mu_j)\leq p_{j,\tau^*_j(x)}^s(x,\mu_j)$ for all $\mu_j\in\mathbb{R}$.
        \end{itemize}
        If (a) and (c) hold, and the inequality in (b) is not always strict, we say the \emph{evidence in $x$ and $x^{\prime}$ is the same}.
    \end{definition}

\section{Compatible simultaneous confidence bounds for GSD using graphical approaches}\label{SubsectionCompatibleBounds}
In this section, we present simultaneous confidence bounds for the repeated (i.e. $\lambda=r$) and sequential (i.e. $\lambda=s$) graphical Algorithm~\ref{AlgorithmusMauerBretz} and our suggested efficient multiple adjustment Algorithm~\ref{alg_compromise_compatible} that are compatible. They are based on the bounds from \citealp{SB08} for closure tests. After that, we will introduce informative bounds. 

In the following, we use the definitions from Algorithm~\ref{AlgorithmusMauerBretz}: We consider $1\leq k\leq K$ stages, $\tau^*_j$ is defined by the stage where the data collection for $H_j$ stops, and $\tau^*=\max_{1\leq j\leq m}\{\tau^*_j\}\leq K$ gives the stage where the whole trial is terminated. For the current considered stage $1\leq k\leq K$ the last stage of data collection for hypothesis $H_j$ is $k_j^*=k_j^*(k)=\min\{\tau^*_j,k\}$. Note that $\tau^*$, $\tau^*_j$ and $k_j^*$, $1\leq j\leq m$, depend all on the sample and decisions. The hypotheses rejected by Algorithm~\ref{AlgorithmusMauerBretz} up to stage $k$ are defined by $\bar R_k^{\lambda}$. 

We start with compatible bounds from \citealp{SB08} for Algorithm \ref{AlgorithmusMauerBretz}. We take over the overline-notation from \citealp{SB08} again to distinguish them from our informative bounds that will be introduced in Section~\ref{SubsectionInformativeBounds}. Let $\lambda=\{r,s\}$. In addition to the inputs for Algorithm \ref{AlgorithmusMauerBretz}, significance levels $\alpha_j(\varnothing)$, $1\leq j\leq m$, fulfilling $\sum_{j=1}^m\alpha_j(\varnothing)=\alpha$ need to be chosen and fixed in advance. They come into play in the case that all hypotheses can be rejected. For each stage $1\leq k\leq K$,  simultaneous confidence bounds are given by
\begin{align}
    \bar{L}_{j,k}^{\lambda}:=
    \begin{cases}
			0 & \text{if $j\in \bar R_k^{\lambda}, \bar R_k^{\lambda}\neq I$}\\
            (p_{j,k_j^*}^{\lambda})^{-1}(\alpha_j(I\setminus \bar R_k^{\lambda})) & \text{if $j\notin \bar R_k^{\lambda}$}\\
            \max\{0,(p_{j,k_j^*}^{\lambda})^{-1}(\alpha_j(\varnothing))\}& \text{if $\bar R_k^{\lambda}=I$}
		 \end{cases},\quad 1\leq j\leq m,\quad k\leq \tau^*,\label{compatibleBounds}
\end{align}
and $\bar{L}_{j,k}^{\lambda}:=\bar{L}_{j,\tau^*}^{\lambda}$ for all $\tau^*<k\leq K$. The bounds  $\bar{L}_{j,k}^{\lambda}$ exactly reflect the test decisions made by Algorithm~\ref{AlgorithmusMauerBretz} up to stage $k$, $1\leq k\leq\tau^*$. In particular, the final confidence bounds $\bar{L}_{j,\tau^*}^{\lambda}$ are compatible with the graphical group sequential test, i.e. they meet $\bar{L}_{j,\tau^*}^{\lambda}\geq 0$ if and only if $H_j$ is rejected at the end of  Algorithm~\ref{AlgorithmusMauerBretz}. 
For each $1\leq k\leq K$, the coverage probability of the SCI is at least $1-\alpha$ because the bounds are more conservative than the compatible bounds corresponding to the graphical test using the p-values $p_{j,K}^s$, $j\in I$, for which the coverage probability is guaranteed by \citealp{SB08}. 

It is easy to see that the sequential SCI bounds fulfil $\bar L_{j,k}^s\leq \bar L_{j,k+1}^s$ for all $j\in I$ and $1\leq k <K$. This follows from the definition of the sequential p-values, which yields increasing sequential p-values along the stages. 

We now give compatible bounds for the efficient multiple adjustment Algorithm \ref{alg_compromise_compatible}. Of course, one could perform Algorithm \ref{alg_compromise_compatible} after each stage $k$ of Algorithm \ref{AlgorithmusMauerBretz} for $\lambda=s$. Then one obtains compatible bounds at stage $k$ by replacing in the following $\tau^*$ by $k$, $\tau_j^*$ by $k_j^*(k)$, $\bar R^c$ by $\bar R^c_k$ and $\bar R^s$ by $\bar R^s_k$.

Simultaneous confidence bounds that are compatible with Algorithm \ref{alg_compromise_compatible} are given by

\begin{align}
    \bar{L}_{j,\tau^*}^{c}:=
    \begin{cases}
			0 & \text{if $j\in \bar R^{c}, \bar R^{s}\neq I$}\\
            (p_{j,\tau_j^*}^{r})^{-1}(\alpha_j(\{j\}\cup I\setminus \bar R^{s})) & \text{if $j\notin \bar R^{c}$}\\
            \max\{0,(p_{j,\tau_j^*}^{r})^{-1}(\alpha_j(\varnothing))\}& \text{if $j\in\bar R^c,  \bar R^{s}=I$}
		 \end{cases}~, \quad 1\leq j\leq m.\label{mixtureCompatibleBounds}
\end{align}
Similar to the control of the FWER of Algorithm \ref{alg_compromise_compatible}, the coverage probability of the SCI defined by (\ref{mixtureCompatibleBounds}) is at leat $1-\alpha$ because of $\bar L_{j,\tau^*}^c\leq\bar L_{j,\tau^*}^s$, for all $1\leq j\leq m$.

For the compatible bounds (\ref{compatibleBounds}) based on \citealp{bretz2009graphical}, it is well known and obvious that they do not meet the informativeness definition \ref{definitionInformativität}  because the bounds remain stuck at the border of the null hypothesis as long as not all null hypotheses have been rejected by the corresponding compatible test. This applies in a similar way to the bounds (\ref{mixtureCompatibleBounds}). However, the difference here is that the boundaries remain stuck at the border of the null hypothesis, as long as the sequential variant of Algorithm \ref{AlgorithmusMauerBretz} does not reject all hypotheses. This is not the corresponding compatible test.

\section{Informative simultaneous confidence bounds}\label{mainChapterInformativeSCIs}
In this section, we first give a short overview of the method to construct informative simultaneous confidence bounds for graphical single-stage test procedures from \citealp{informativeSCIBrannathKlugeScharpenberg}. After that, the algorithm will be extended to yield iterative algorithms for calculating these bounds for the graphical group sequential test procedures introduced in Section~\ref{subsectionGraphicalGSD}.
\subsection{Characteristics of the method}\label{sectionCharacteristics}
The method from \citealp{informativeSCIBrannathKlugeScharpenberg} gives informative simultaneous confidence bounds for all graphical single-stage tests from \citealp{bretz2009graphical}. The key idea is to modify the local levels of the graph (i.e., the levels of the underlying closure test) in a way that roughly speaking, some significance level is retained after rejecting a hypothesis $H_j=H_j^0$, $j=1,\dots,m$, in order to test the shifted hypotheses $H_j^{\mu_j}$, $\mu_j> 0$, and increase the confidence bound. That means not the whole level is passed to the other hypotheses after rejection. The strategy is formalized by defining for each $\mu=(\mu_1,\dots,\mu_m)\in\mathbb{R}^m$, a weighted Bonferroni Test (i.e. levels $\alpha_j^{\mu}$, $j=1,\dots, m$, fulfilling $\sum_{j=1}^m\alpha_j^{\mu}=\alpha$) for the intersection of shifted hypotheses $H^{\mu}:=H_1^{\mu_1}\cap\dots\cap H_m^{\mu_m}$. This hypothesis is rejected if and only if $p_j(\mu_j)\leq\alpha_j^{\mu}$ for at least one $j$ holds true. Here, $p_j(\mu_j)$ describes the (local) p-value of the first and at the same time final stage. An m-dimensional confidence set is then defined via 
\begin{align}
    C:=C(x):=\{\mu\in\mathbb{R}^m: H^{\mu}\text{~is not rejected}\}=\{\mu\in\mathbb{R}^m:\min_{\substack{j=1,\dots,m,\\ \alpha_j^{\mu}>0}}p_j(\mu_j,x)/\alpha_j^{\mu}>1 \}\label{confidenceRegionInformativeClassical}~.
\end{align}
It has a coverage probability of at least $1-\alpha$ because the local Bonferroni tests maintain the significance level $\alpha$. The informative bounds are then defined as the component-wise projection of $C$, i.e., the smallest (left-sided limited) rectangle $L$ such that $L=(L_1,\dots, L_m)\supseteq C$ holds true. Because $L$ always contains $C$, it has a coverage probability of at least $1-\alpha$.     

The local significance levels are defined using dual graphs of the original graphical test. This is explained in detail in \citealp{informativeSCIBrannathKlugeScharpenberg}. We will give a short summary using our example of a hierarchical test with four hypotheses: The \emph{first step} is the creation of a dual graph.  For the dual graphs, an information weight $q\in (0,1)$ must be chosen and specified in advance. Roughly speaking, the information weight describes how much level of a rejected hypothesis is kept for testing corresponding shifted hypotheses and how much is passed to the other hypotheses. Let now be $\mu=(\mu_1,\dots,\mu_m)\in\mathbb{R}^m$. In our example, we have $m=4,$ and we assume that $\mu_j>0$ for $j=1,2,4$ and $\mu_3\leq 0$. The starting point of the dual graph $G^{\mu}$ is the underlying graph $G$, see Figure~\ref{PictureHTest}. Now, for all $1\leq j\leq m$ with $\mu_j\leq 0$ all path starting at $H_j=H_j^0$ are deleted and $H_j$ is replaced by the shifted hypotheses $H_j=H_j^{\mu_j}:\theta_j\leq\mu_j$. In our example, this applies to the third component, see Figure~\ref{PictureHTestDualGraph}. For all $j$ with $\mu_j>0$, the node $H_j$ is kept, and an additional node for $H_j^{\mu_j}$ is added with initial local level $0$. Additionally, for all those $j$ an arrow is added that goes from $H_j$ to $H_j^{\mu_j}$ with transition weight $q^{\mu_j}$. The transition weights $g_{ji}$ of all other arrows going from $H_j$ to a $H_i=H_i^0$ (if $\mu_i>0$) or to a $H_i^{\mu_i}$ (if $\mu_i\leq 0$), $i\neq j$, are changed to $g_{ji}(1-q^{\mu_j})$. In our example, this is done for the hypotheses $H_1$ and $H_2$ because $\mu_1>0$ and $\mu_2>0$. For the fourth hypothesis, it was also assumed that $\mu_4>0$. Because in the original graph, there is no outgoing arrow, which is called in \citealp{informativeSCIBrannathKlugeScharpenberg} a ``non-complete graph'', the transition weight of the arrow from $H_4$ to $H_4^{\mu_4}$ can be set to $1$. The resulting dual graph $G^{\mu}$ is a graphical test in the sense of \citealp{bretz2009graphical} because when adding the shifted hypotheses, it is ensured that the sum of the outgoing arrows only adds up to 1. 

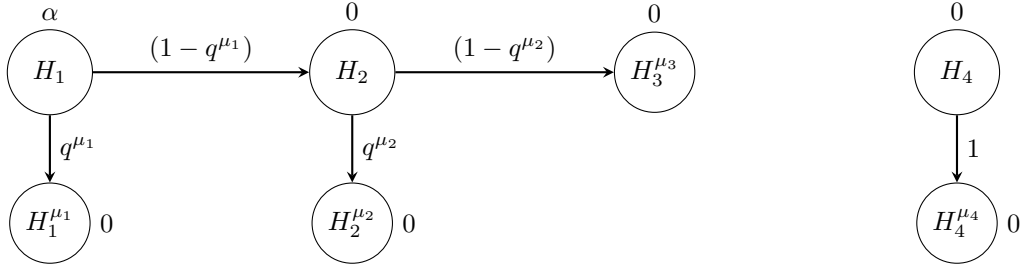
\begin{figure}[h]
\begin{center}
\begin{tikzpicture} [node distance = 4cm, on grid, auto]

\node[label={$\alpha$}, inner sep = 6pt] (H1) [state] {$H_1$};
\node[label={$0$}, inner sep = 6pt] (H2) [state, right = of H1] {$H_2$};
\node[label={$0$}] (H3) [state, right = of H2] {$H_3^{\mu_3}$};
\node[label={$0$}, inner sep = 6pt] (H4) [state, right = of H3] {$H_4$};
\node[label=right:{$0$}] (H1mod) at (0,-2) [state] {$H_1^{\mu_1}$};
\node[label=right:{$0$}] (H2mod) at (4,-2) [state] {$H_2^{\mu_2}$};
\node[label=right:{$0$}] (H4mod) at (12,-2) [state] {$H_4^{\mu_4}$};

\path [-stealth, thick]
	(H1) edge  node {$(1-q^{\mu_1})$}    (H2)
	(H2) edge  node {$(1-q^{\mu_2})$}    (H3)
	(H1) edge  node {$q^{\mu_1}$} (H1mod)
	(H2) edge  node {$q^{\mu_2}$} (H2mod)
	(H4) edge  node {$1$} (H4mod);
	
\end{tikzpicture}
\end{center}\caption{Dual graph for the hierarchical test with 4 hypotheses for a $\mu=(\mu_1,\mu_2,\mu_3,\mu_4)\in\mathbb{R}_{-\infty}^4$, with $\mu_1>0$, $\mu_2>0$, $\mu_3\leq 0$ and $\mu_4>0$.}\label{PictureHTestDualGraph}
\end{figure}

When looking at the dual graph from Figure~\ref{PictureHTestDualGraph} and the underlying original one from Figure~\ref{PictureHTest}, the role of the information weight becomes clear. The smaller $q$ is, the closer the dual graph comes to the original graph, and the larger $q$ is, the more weight is given to the shifted hypothesis after rejection. Note that all shifted hypotheses $H_j^{\mu_j}$ are contained in the dual graph. Thus, the graph can be understood as a test for $H^{\mu}=\cap_{j=1}^m H_j^{\mu_j}$. In general, the test of an intersection is included in a graph in the following way: All hypotheses except the ones that contribute to the intersection hypotheses are rejected, and the graphical test update algorithm is performed. The test for the intersection hypotheses is then given by the Bonferroni test, based on the actual local levels. This is exactly what is done in the \emph{second step} after creating the dual graph: the graphical test update algorithm by \citealp{bretz2009graphical} is performed, assuming that all $H_j=H_j^{0}$ (in the case of the hierarchical test $H_1$, $H_2$, and $H_4$) that are still contained in the graph can be rejected. This gives the local significance levels $\alpha_j^{\mu}$ used in (\ref{confidenceRegionInformativeClassical}). We want to summarize the most important properties of them mentioned and proven in \citealp{informativeSCIBrannathKlugeScharpenberg}.
\begin{proposition}\label{propositionWithImportantProperties}
For all $\mu=(\mu_1,\dots,\mu_m)\in\mathbb{R}^m$ it holds 
\begin{itemize}
    \item[(i)] the local level always satisfy $\sum_{j=1}^m \alpha_j^{\mu}=\alpha$;
    \item[(ii)] the local levels are of the form $\alpha_j^{\mu}=q^{\mu_j\vee 0}\nu_j(\mu)\alpha$ where $\nu_j:\mathbb{R}^m\rightarrow\mathbb{R}_{\geq 0}$ is a function that is continuous and non-decreasing in each component and $q^{(\mu_j\vee 0)}\nu_j(\mu)\leq 1$;
    \item[(iii)] for all $\mu\in\mathbb{R}^m$ such that $\mu_j\leq 0$ the local level fulfils $\alpha_j^{\mu}\geq\alpha_j$, $1\leq j\leq m$.
\end{itemize}
\end{proposition}
These characteristics played an important role in \citealp{informativeSCIBrannathKlugeScharpenberg} for the construction of an algorithm that converges (from below) against the (projected) bounds $L$. It was also shown in \citealp{BuchWassmerBrannath2025} that the bounds $L$ meet the informativeness Definition \ref{definitionInformativität}. In Section~\ref{SubsectionInformativeBounds}, we will extend this algorithm to graphical group sequential test procedures.

\subsection{Informative simultaneous confidence bounds for GSD using graphical approaches} \label{SubsectionInformativeBounds}
In this section, we will first extend the algorithm from \citealp{informativeSCIBrannathKlugeScharpenberg} to calculate informative bounds at each stage $1\leq k\leq K$ of the trial. This will be called \emph{primary algorithm}. After that, we will formulate algorithms that yield informative simultaneous confidence bounds for the group sequential test procedures from Section~\ref{subsectionGraphicalGSD} by repeatedly calling the primary algorithm.

As in Section~\ref{sectionCharacteristics}, we have a fixed underlying graphical test with $K\in\mathbb{N}$ stages for $m$ hypotheses and a fixed information weight $q\in (0,1)$. We will work again with the local significance levels $\alpha_j^{\mu}$, $1\leq j\leq m$, $\mu\in\mathbb{R}^m$, uniquely defined by the dual graphs. For $\lambda\in\{r,s\}$ and $1\leq k\leq K$ we define the confidence set
\begin{align}
    C^{\lambda}_{k}:=C_k^{\lambda}(x):=\{\mu\in\rk: \min_{\substack{j=1,\dots,m,\\ \alpha_j^{\mu}>0}} p^{\lambda}_{j,k^*_j(k)}(\mu_j,x)/\alpha_j^{\mu}> 1\}~.\label{confidenceSetEachStage}
\end{align}
As in the previous sections, it holds $k_j^{*}(k)=\min\{\tau_j^{*},k\}$ where $\tau_j^*$, $1\leq\tau_j^*\leq K$, is the final stage of data collection for the hypothesis $H_j$, $1\leq j\leq m$. The stage $k_j^{*}(k)$ depends on the data and decisions. The confidence set is similar to the one from (\ref{confidenceRegionInformativeClassical}), whereby the p-values are now given by group sequential p-values. Note that it always holds for all $\mu_j\in\mathbb{R}$ and for all observed data $x$ and decisions that $p_{j,k_j^*(k)}^{\lambda}(\mu_j,x)\leq p_{j,K}^{s}(\mu_j,x)$. This guarantees that we have valid p-values. In particular, the confidence set (\ref{confidenceSetEachStage}) has a coverage probability of at least $1-\alpha$. We denote the corresponding projected (left-sided limited) SCI in the following by $L_k^{\lambda}=(L_{1,k}^{\lambda},\dots,L_{m,k}^{\lambda})$.

With the Primary Algorithm \ref{alg_1} a lower approximation of $L_k^{\lambda}=(L_{1,k}^{\lambda},\dots,L_{m,k}^{\lambda})$ can be calculated. The algorithm also includes the calculation of an upper estimate, which was not included in previous work, such as \citealp{informativeSCIBrannathKlugeScharpenberg}.

\begin{algorithm}[h]
\caption{Primary algorithm}\label{alg_1}
\KwInput{Graphical group sequential test with $K$ stages for $m$ hypotheses, stage $k$, $1\leq k\leq K$, of the trial, for each hypothesis last performed stage $k_j^*=k_j^*(k)$, $1\leq k_j^*\leq K$, and p-values $p_{j,k_j^*}^{\lambda}$, $1\leq j\leq m$, for a fixed $\lambda\in\{r,s\}$, information weight $q$, $\varepsilon$, significance level $\alpha$ and strictly decreasing, strictly positive sequence $(\delta^{(\ell)})_{\ell\in\mathbb{N}_0}$ that converges to $0$ and fulfills $\alpha+\delta^{(0)}<1$.}
Set $\ell=0$\;
Find lower starting vector $\mu^{(0)}\in\rk$ such that
\begin{align}
    p_{j,k_j^*}^{\lambda}(\mu_j^{(0)})\leq q^{(\mu_j^{(0)}\vee 0)}\nu_j(\mu^{(0)})\alpha \quad \text{for all } j=1,\dots,m.\label{eq_astartLower}
\end{align}

Find upper starting vector $\rho^{(0)}\in\rk$ such that
\begin{align}
    p_{j,k_j^*}^{\lambda}(\rho_j^{(0)})\geq q^{(\rho_j^{(0)}\vee 0)} \nu_j(\rho^{(0)})(\alpha+\delta^{(0)})\quad\text{for all } j=1,\dots,m.\label{eq_astartUpper}
\end{align}

\While{$\|\mu^{(\ell)} - \rho^{(\ell)}\|_2 \geq \varepsilon$}{
Calculate $\mu^{(\ell+1)}$ such that
\begin{align}
    p_{j,k_j^*}^{\lambda}(\mu_j^{(\ell+1)})=q^{(\mu_j^{(\ell+1)}\vee 0)}\nu_j(\mu^{(\ell)})\alpha \quad\text{for all } j=1,\dots,m. \label{eq_stepLower}
\end{align}

Calculate $\rho^{(\ell+1)}$ such that
\begin{align}
  p_{j,k_j^*}^{\lambda}(\rho_{j}^{(\ell+1)})=q^{(\rho_{j}^{(\ell+1)}\vee 0)}\nu_j(\rho^{(\ell)})(\alpha+\delta^{(\ell)}) \quad\text{for all } j=1,\dots,m.\label{eq_stepUpper}
\end{align}
Update $\ell\rightarrow\ell +1$\;
}
\Return{$\mu^{(s)}$, $\rho^{(s)}$ and $\|\mu^{(s)} - \rho^{(s)}\|_2$ where $s=\ell-1$ is the last index where the approximations where calculated.}
\end{algorithm}
The following Theorem summarizes the most important properties of Algorithm \ref{alg_1}.

\begin{theorem}\label{theoremMainResults}
The following properties are satisfied by Algorithm \ref{alg_1}.
\begin{itemize}
    \item[(a)] For each valid starting vector $\mu^{(0)}\in\mathbb{R}^m$ such that (\ref{eq_astartLower}) is met, the sequence $(\mu^{(\ell)})_{\ell\in\mathbb{N}_0}$ is component-wise increasing and converges from below against $L_k^{\lambda}$.
    \item[(b)] For each valid starting vector $\rho^{(0)}\in\mathbb{R}^m$ such that (\ref{eq_astartUpper}) is met, the sequence $(\rho^{(\ell)})_{\ell\in\mathbb{N}_0}$ is component-wise decreasing and converges from above against $L_k^{\lambda}$.
    \item[(c)] The SCI bounds $L_k^{\lambda}$ are \emph{informative}. 
\end{itemize}
\end{theorem}
Properties (a) and (c) are proven in \citealp{informativeSCIBrannathKlugeScharpenberg}. For a proof of property (b), we refer to the Appendix. Theorem \ref{theoremMainResults} states that the limit vector of the algorithm is unique regardless of the starting vector.

In the following, we would like to take a closer look at the input vectors and returned vectors of Algorithm \ref{alg_1}. 
\begin{remark}
Valid starting vectors $\mu^{(0)}$ and $\rho^{(0)}$ for calculating a lower and upper approximation of $L_k^{\lambda}$ are given by 
\begin{align}
    \mu_j^{(0)}=\min\{0, (p_{j,k_j^*}^{\lambda})^{-1}(\alpha_j)\} \quad \text{and}\quad \rho_j^{(0)}=(p_{j,k_j^*}^{\lambda})^{-1}(\alpha +\delta^{(0)}), \quad 1\leq j\leq m~.\label{startVectors}
\end{align}
The vector $\mu^{(0)}$ meets equation (\ref{eq_astartLower}) because $\mu_j^{(0)}\leq 0$ implies $\alpha_j\leq\alpha_j^{\mu}=q^{(\mu_j^{(0)}\vee 0)}\nu_j(\mu^{(0)})\alpha$. The vector $\rho^{(0)}$ meets equation (\ref{eq_astartUpper}) because $1\geq q^{(\rho_j^{(0)}\vee 0)}\nu_j(\rho^{(0)})$.
\end{remark}

For the returned vectors it holds $\mu^{(s)} \leq L_k^{\lambda}\leq \rho^{(s)}$ pointwise. It is to be expected that the lower approximation is of greater practical interest because the corresponding SCI has a coverage probability of at least $1-\alpha$ because it contains $(L_{1,k}^{\lambda},\infty)\times\dots\times(L_{m,k}^{\lambda},\infty)$ which in turn contains $C_k^{\lambda}$. This does not, in general, hold for the upper approximation. However, the upper approximation plays an important role in estimating how close the lower approximation is to the actual limit $L_k^{\lambda}$ and is therefore also of practical interest. By having upper approximations, we are able to change the stop criterion of Algorithm \ref{alg_1} compared to the one from \citealp{informativeSCIBrannathKlugeScharpenberg}, where the calculation of the lower approximation is stopped if several consecutive calculated approximations are ``close'' to each other.  We will add the upper approximation and the new stop criterion in our \textsf{R}-package \emph{informativeSCI}, which is available from \textsf{CRAN} or \textsf{github} and can be used to calculate the informative SCI bounds. 

Before proceeding with building informative variants of Algorithm \ref{AlgorithmusMauerBretz} by repeatedly using the Primary Algorithm \ref{alg_1}, we come back to a monotonicity property for $\lambda=s$. For all $1\leq j\leq m$ and $1\leq k <K$ it holds by definition $k_j^*(k)\leq k_j^*(k+1)$. Using the definition from the sequential p-values, we can conclude that for all $\mu_j\in\mathbb{R}$ and data and decisions, we have $p_{j,k_j^*(k+1)}^s(\mu_j,x)\leq p_{j,k_j^*(k+1)}^s(\mu_j,x)$. This implies $C_k^s\supseteq C_{k+1}^s$ that yields increasing projected bounds $L_k^{s}\leq L_{k+1}^s$ component-wise. The same holds as already stated for the compatible bounds, i.e. $\bar L_{k}^s\leq \bar L_{k+1}^s$ component-wise. 

In \citealp{informativeSCIBrannathKlugeScharpenberg}, it was demonstrated that if a graphical single-stage test is given and the informative simultaneous confidence intervals are constructed using the graphical test as basis for the dual graphs and the same p-values $p_j(\mu_j)$, $\mu_j\in\mathbb{R}$, $j\in\{1,\dots,m\}$, that the test decisions made by the test induced by the informative bounds become more and more similar to those made by the underlying graphical test if $q$ approaches $0$.

In the graphical group sequential test procedures from Algorithm~\ref{AlgorithmusMauerBretz}, the p-values are ``updated'' after each stage $k$ of data collection. To obtain corresponding informative simultaneous confidence intervals where the test decisions become again more and more similar for $q$ approaching $0$, in the following Algorithm \ref{alg_repeated}, the p-values should be updated ``in the same way''. That means after each stage, the index set $D$ should be updated as in Algorithm \ref{AlgorithmusMauerBretz}. The algorithm can be executed for $\lambda\in\{r,s\}$. In Section~\ref{subsectionLambda_r} and Section~\ref{SubsectionLambda_s}, we will come to this in detail. 

\begin{algorithm}[h]
\setstretch{1.2}
\caption{ISCI for graphical GSD with $\lambda\in\{r,s\}$} \label{alg_repeated}
\KwInput{Graphical group sequential test with $K$ stages for $m$ hypotheses, information weight $q$, overall significance level $\alpha$ and nominal significance levels $\alpha_{j,k}^*(\gamma)$, $\gamma\in[0,1]$, $1\leq k\leq K$, $1\leq j\leq m$.}
Set $D=\{1,\dots,m\}$, $k=1$, and for all $1\leq j\leq m$ set $k_j^*=1$\;

\While{$|D|\geq 1$ \emph{\textbf{and}} $k\leq K$}{
Continue the trial with the hypotheses $H_j$ for $j\in D$ and compute local p-values $p_{j,k}$ and group sequential p-values $p_{j,k}^{\lambda}$\;
Run Algorithm~\ref{alg_1} with $(k_j^*)_{1\leq j\leq m}$ and $(p_{j,k_j^*}^{\lambda})_{1\leq j\leq m}$ to obtain $L^{\lambda}_{k}=(L_{1,k}^{\lambda},\dots L_{m,k}^{\lambda})$\;
Set $R_k^{\lambda}:=\{1\leq j\leq m : L_{j,k}^{\lambda}\geq 0\}$\;
Decide for each $j\in D$ whether to stop data collection and update in this case $D\rightarrow D\setminus\{j\}$ and set $\tau^*_j=k$\;
Update $k\rightarrow k+1$\;
Update for all $j\in D$ the current  stage $k_j^*\rightarrow  k$\;}
Let $\tau^*=\max_{1\leq j\leq m}\{\tau^*_{j}\}=k$  be the stage where the trial is terminated\;
\Return{$L^{\lambda} :=L_{\tau^*}^{\lambda}=(L_{1,\tau^*}^{\lambda},\dots,L_{m,\tau^*}^{\lambda})$ and index set $R^{\lambda}:=R_{\tau^*}^{\lambda}$ of rejected hypotheses.}
\end{algorithm}
Algorithm~\ref{alg_repeated} calculates (for $\lambda\in\{r,s\}$) the left-sided bounds of the smallest SCI containing 
\begin{align}
    C^{\lambda}:=C_{\tau^*}^{\lambda}:= C_{\tau^*}^{\lambda}(x):=\{\mu\in\rk: \min_{\substack{j=1,\dots,m,\\ \alpha_j^{\mu}>0}} p^{\lambda}_{j,\tau^*_j(\tau^*)}(\mu_j,x)/\alpha_j^{\mu}> 1\}~.\label{confidenceSetRepeatlyCallingPrimary}
\end{align}
Remember that the stage $\tau^*$ where the trial is terminated as well as $\tau_j^{*}$, $1\leq j\leq m$, depend on the data and decisions. As for the confidence sets from (\ref{confidenceSetEachStage}) for a stage $1\leq k<\tau^*$, we have that $p_{j,\tau_j^*}^{\lambda}(\mu_j)\geq p_{j,K}^{s}(\mu_j)$, $\mu_j\in\mathbb{R}$. This guarantees that (\ref{confidenceSetRepeatlyCallingPrimary}) is a confidence set of level $1-\alpha$. 

The \emph{informativeness} of $L_{\tau^*}^{\lambda}$ obtained by Algorithm~\ref{alg_repeated} in the sense of Definition~\ref{definitionInformativität} under the understanding of evidence from Definition~\ref{DefinitionEvidenzBretzAlgo} (set $k=\tau^*$) follows directly from Theorem~\ref{theoremMainResults}~(c).

In the next two sections, we give details on the $\lambda=r$ and $\lambda=s$ variant of Algorithm~\ref{alg_repeated} including valid starting vectors.
\subsubsection{Algorithm \ref{alg_1} for $\lambda=r$}\label{subsectionLambda_r}
In each stage $k$, $1\leq k\leq\tau^*$, Algorithm~\ref{alg_repeated} calls the Primary Algorithm~\ref{alg_1} for which starting vectors for the lower and upper approximation of $L_k^{\lambda}$ are needed. In (\ref{startVectors}) possible starting vectors are listed. However, for the lower approximation under certain assumptions the results from previous stages can be used to start the Primary Algorithm~\ref{alg_1} with a starting vector that is component-wise greater than the one from (\ref{startVectors}). By this, the number of numerical calculations can be reduced. This is related to the following problem: For Algorithm~\ref{AlgorithmusMauerBretz} with $\lambda=r$ we have argued in Section~\ref{subsectionGraphicalGSD} that continuing the data collection for a specific component despite rejection at a certain stage may lead to inconclusive rejections of hypotheses. The same holds for the $\lambda=r$ variant of Algorithm~\ref{alg_repeated}. Under similar assumptions as in Section~\ref{subsectionGraphicalGSD}, this can be prevented as stated in the following lemma.
\begin{lemma}\label{LemmaMonotonicityLambdaR}
Let $\lambda=r$ in Algorithm \ref{alg_repeated} and assume for all  $j\in\{1,\dots,m\}$ with $j\in\cup_{k=1}^{\tau^*}R_k^r$ that it holds $\tau_j^*\leq r_j:=\min\{1\leq k\leq\tau^*~|~j\in R_{k}^r\}$, i.e. if a hypothesis is rejected the first time at a certain stage the data collection for this hypothesis is stopped or has already been stopped at an earlier stage. Then, the index sets of rejected hypotheses by Algorithm \ref{alg_repeated} fulfil $R_{k}^r\subseteq R_{k+1}^r$, $1\leq k<\tau^*$. In particular, for all $j\in R_{\tau^*}^{r}$ we have for all $k$ with $r_j\leq k <\tau^*$ that $L_{j,k}^{r}\leq L_{j,k+1}^{r}$, i.e. the confidence bounds increase in these components.
\end{lemma}
For a proof we refer to the Appendix. In particular, it can be seen from this proof that a valid starting vector at stage $2\leq k\leq\tau^*$ for the lower approximation is given by
\begin{align*}
    \mu_j^{(0)}=\begin{cases}
             \min\{0, (p_{j,k_j^*(k)}^{r})^{-1}(\alpha_j)\}~, & j\notin R_{k-1}^r\\
            L_{j,k-1}^r~, & j\in R_{k-1}^r
		 \end{cases}~,\quad 1\leq j\leq m~.
\end{align*}
For the upper approximation, a valid starting vector is given in (\ref{startVectors}).
If the assumptions of Lemma~\ref{LemmaMonotonicityLambdaR} are not fulfilled, the index sets of rejected hypotheses are in general not nested across the stages. That means that rejections of hypotheses may be revoked at a later stage. The final test decisions at stage $\tau^*$ are made by Algorithm~\ref{alg_1} for all hypotheses $H_j$ using all the evidence up to the last stage of data collection $\tau_j^*$, $1\leq j\leq m$. In particular, Algorithm~\ref{alg_1} provides an informative variant of the conservative testing strategy mentioned in Section~\ref{subsectionGraphicalGSD} for our hierarchical test example from Table~\ref{tab:hiertest4Hyp}. The hierarchical test is executed stage-wise. A corresponding informative variant is obtained by Algorithm~\ref{alg_1} when, after the first stage, it remains $D=\{1,2,3,4\}$.



\subsubsection{Algorithm \ref{alg_1} for $\lambda=s$}\label{SubsectionLambda_s}
Compared to Section~\ref{subsectionLambda_r} the simultaneous confidence bounds $L_k^s$, $1\leq k\leq\tau^*$ are component-wise increasing along the stages as argued earlier. In particular it always holds for $1\leq k<\tau^*$ that $R_k^s\subseteq R_{k+1}^s$. In the Appendix it is also formally shown that at stage $2\leq k\leq \tau^*$ a valid starting vector for the lower approximation is given by
\begin{align*}
    \mu^{(0)}=L_{k-1}^s~.
\end{align*}
For the upper approximation, a valid starting vector is again given in (\ref{startVectors}).

\subsubsection{Informative SCI for graphical GSD with efficient multiple adjustment}
We now come back to the aim of defining informative SCI bounds and a corresponding powerful induced test, using for all hypotheses the evidence up to the last or current stage of data collection. As for Algorithm~\ref{alg_compromise_compatible} for reasons of simplicity we will formulate Algorithm~\ref{alg_compromise} as a ``follow-up algorithm'' to Algorithm~\ref{alg_1} with $\lambda=s$. Of course, the follow-up algorithm can also be executed after each stage $k$, $1\leq k\leq\tau^*$, of Algorithm~\ref{alg_1} with $\lambda=s$. The Algorithm~\ref{alg_compromise} can be interpreted as an ``informative'' version of Algorithm~\ref{alg_compromise_compatible}.
\begin{algorithm}[h]
\setstretch{1.2}
\caption{ISCI for graphical GSD with efficient multiple adjustment} \label{alg_compromise}
\KwInput{Graphical group sequential test with $K$ stages for $m$ hypotheses, lower confidence bounds $L^s=L_{\tau^*}^s=(L_{1,\tau^*}^s,\dots,L_{m,\tau^*}^s)$ obtained from Algorithm~\ref{alg_repeated} for $\lambda=s$ with corresponding last conducted stages $\tau^*_j=\tau^*_j(x)$ for each hypothesis $H_j$, $1\leq j\leq m$, and p-values $p_{j,\tau^*_j}^r$, $1\leq j\leq m$, information weight $q$, significance level $\alpha$.}

\For{$j=1$ \emph{\textbf{to}} $m$}{
Define $\tilde\omega_j(\mu_j):=q^{(\mu_j\vee0)}\nu_j((L_{1,\tau^*}^s,\dots,L_{j-1,\tau^*}^s,\mu_j,L_{j+1,\tau^*}^s,\dots,L_{m,\tau^*}^s))$, $\mu_j\in\mathbb{R}$\;
\eIf{$\tilde{\omega}_j(0)\alpha>0$}{
\vspace{0.3em}
Perform a bisection search to find the unique $L_{j,\tau^*}^c\leq L_{j,\tau^*}^s$ such that
\begin{align}
    & p_{j,\tau^*_j}^r(L_{j,\tau^*}^c)= \tilde\omega_j(L_{j,\tau^*}^c)\alpha. \label{keyEquationMixtureInformative}
\end{align}}{
Set $L_{j,\tau^*}^c:=-\infty$\;
}
}
Set $R^c:=\{j\in\{1,\dots,m\}: L_{j,\tau^*}^c\geq 0\}$\;
\Return{$L^c:=L_{\tau^*}^c=(L_{1,\tau^*}^c,\dots,L_{m,\tau^*}^c)$ and index set $R^c$ of rejected hypotheses.}
\end{algorithm}
First, Algorithm~\ref{alg_compromise} calculates the bounds $L_{\tau^*}^s$ by calling Algorithm~\ref{alg_1} for $\lambda=s$. For each hypothesis $H_j$, $1\leq j\leq m$, holds: The larger each $L_{i,\tau^*}^s$, $i\in\{1,\dots,m\}$, $i\neq j$, the larger is $\tilde{\omega}_j(\mu_j)$, $\mu_j\in\mathbb{R}$. Roughly speaking, the evidence from earlier stages (i.e., higher bounds $L_{i,\tau^*}^s$, $i\neq j$) is used to test $H_j$ at a higher level and finally obtain larger bounds $L_{j,\tau^*}^c$. In the Appendix we argue that for each component $j\in\{1,\dots,m\}$ there always exists a \emph{unique} $L_{j,\tau^*}^c\in\mathbb{R}_{-\infty}$, $L_{j,\tau^*}^c\leq L_{j,\tau^*}^s$, fulfilling (\ref{keyEquationMixtureInformative}). Additionally, it is argued that, for all components $j\in\{1,\dots, m\}$ with $\tilde{\omega}_j(0)\alpha>0$, the bisection search from Algorithm~\ref{alg_compromise} can be performed in practice because there exists a real-value search area given below. In particular, $L_{j,\tau^*}^c>-\infty$ holds in this case.
\begin{remark}\label{remarkRealValuedSearchAreaBisectionMixture}
    For all components $j\in\{1,\dots, m\}$ from Algorithm~\ref{alg_compromise} with $\tilde{\omega}_j(0)\alpha>0$, the bound $L_{j,\tau^*}^c$ is contained in 
    \begin{align}
        L_{j,\tau^*}^c\in\left[(p_{j,\tau_j^*}^r)^{-1}(\tilde\omega_j(0)\alpha)\,, L_{j,\tau^*}^s \right]~,\label{startVectorMixtureInformative}
    \end{align}
    where $L_{j,\tau^*}^c>-\infty$.
\end{remark}

It remains to argue the \emph{informativeness} of Algorithm~\ref{alg_compromise}. 
\begin{proposition}\label{PropositionInformativenessMixture}
The simultaneous confidence bounds $L_{\tau^*}^c$ obtained by Algorithm~\ref{alg_compromise} satisfy Definition~\ref{definitionInformativität} with the understanding of evidence from Definition~\ref{DefinitionEvidenzMixture}.
\end{proposition}
The proof of the informativeness is straightforward. A formal proof is given in the Appendix.


\section{Median conservative estimators}\label{SectionMedianConservatEstimators}
This section deals with the aim of constructing estimators of treatment effects for the graphical group sequential test procedures presented in the previous sections. The importance of reliable estimation of treatment effects (with minimal or no bias) in adaptive (and group sequential) trials (with multiple hypotheses) is emphasized in a draft of the \citealp{ICH-E20} guideline. The evaluation of the bias, variability, and the mean squared error of the estimators is also highlighted. 

Estimators that are conservatively biased and thus underestimate the treatment effect can also be considered acceptable.

\begin{definition}[Common median conservative estimator]\label{definitionMedianConservative}
    Let $\theta=(\theta_1,\dots,\theta_m)\in\mathbb{R}^m$ be the true and unknown parameter of interest. An estimator $\hat{\vartheta}=(\hat{\vartheta}_1,\dots,\hat{\vartheta}_m)$ for $\theta$ is called \emph{(common) median conservative} if 
    \begin{align*}
        \mathbb{P}_{\theta}(\hat{\vartheta}\leq\theta)=\mathbb{P}_{\theta}(\cap_{i=1}^m \{\hat{\vartheta}_i\leq\theta_i\})\geq 0.5~.
    \end{align*}
\end{definition}
An estimator that meets Definition~\ref{definitionMedianConservative} can be conservative for two reasons. First, the probability of observing an estimation that is component-wise smaller than $\theta$ can be strictly greater than $0.5$. Second, because of $\mathbb{P}_{\theta}(\hat{\vartheta}_j\leq\theta_j)\geq\mathbb{P}_{\theta}(\cap_{i=1}^m \{\hat{\vartheta}_i\leq\theta_i\})$ it is likely that the estimator is median conservative in the individual components $j\in\{1,\dots,m\}$.

If a significance level of $\alpha_M:=0.5$ is considered and we have a SCI $L=(L_1,\infty)\times\dots\times(L_m,\infty)$ for $\theta$ with coverage probability of at least $1-\alpha_M$ common median conservative estimators can be derived directly from the SCI, namely by $\hat{\vartheta}_j:=L_j$, $j\in\{1,\dots,m\}$. This can be used to construct for all graphical group sequential test procedures presented in this paper (namely Algorithm~\ref{AlgorithmusMauerBretz} for $\lambda\in\{r,s\}$, Algorithm~\ref{alg_compromise_compatible}, Algorithm~\ref{alg_repeated} for $\lambda\in\{r,s\}$ and Algorithm~\ref{alg_compromise}) corresponding (common) median conservative estimators (at each stage $k$, $1\leq k\leq K$). The strategy is the following: First, almost the same underlying graphical test is considered, with the only difference of another overall significance level, namely $\alpha_M:=0.5$ instead of $\alpha$. This means the transition weights $(g_{ij})_{ij}$ remain the same and only the initial local levels are scaled, i.e. $\alpha_j=\alpha_j(I)=\omega_i(I)\cdot\alpha$ is replaced by $\omega_j(I)\cdot\alpha_M$, $I=\{1,\dots,m\}$. Looking at the update rules of the graphs (see, for instance, the ``Update Graph'' part in Algorithm~\ref{AlgorithmusMauerBretz}) it can easily be seen that also the levels $\alpha_j(J)=\omega_i(J)\cdot\alpha$ are replaced by $\alpha_j(J)=\omega_i(J)\cdot\alpha_M$ for all $J\subseteq I$. The p-values $p_{j,k}^r(\mu_j)$ and $p_{j,k}^s(\mu_j)$, $\mu_j\in\mathbb{R}$, $j\in\{1,\dots,m\}$, $1\leq k\leq K$, remain unchanged. Applying this to our hierarchical test introduced in Section~\ref{subsectionGraphicalGSD}, the starting levels $\omega_1(I)\cdot\alpha_M=1\cdot \alpha_M=0.5$ and $\omega_j(I)\cdot\alpha_M=0\cdot0.5=0$ for $j=2,3,4$ are obtained for constructing an SCI. 

The following Lemma\ref{LemmaStrategiesEstimators} summarize for each algorithm introduced in this paper the described strategy for constructing corresponding (common) median conservative estimators. As we will describe later, in most cases, only the estimators based on the informative algorithms are sensible.
\begin{lemma}\label{LemmaStrategiesEstimators} 
For the graphical group sequential tests, one can obtain (common) median conservative strategies by the following procedures:
\begin{itemize}
    \item[(a)] Assume that the graphical GSD from Algorithm~\ref{AlgorithmusMauerBretz} for $\lambda\in\{r,s\}$ and $\alpha$ has been executed up to stage $k$, $1\leq k\leq K$. For corresponding (common) median conservative estimators at stage $k$, perform now the graphical test from \citealp{bretz2009graphical}, i.e., the while-loop part from Algorithm~\ref{AlgorithmusMauerBretz} with the initial levels $\alpha_j(I)=\omega_j(I)\cdot\alpha_M$, $I=\{1,\dots,m\}$, $j\in I$, and the fixed current stages $k_j^{*}(k)$. This gives an index set of rejected hypotheses $\bar{R}_k^{\lambda}(\alpha_M)$. Calculating the bounds from (\ref{compatibleBounds}) with this index set and the scaled levels, (common) median conservative estimators at stage $k$ are obtained by $\hat{\vartheta}_{j,k}:=\bar{L}_{j,k}^{\lambda}(\alpha_M)$, $j\in\{1,\dots,m\}$.
    \item[(b)] Assume that the graphical GSD with efficient multiple adjustment from Algorithm~\ref{alg_compromise_compatible} for the overall level $\alpha$ has been executed (at stage $\tau^*$). For corresponding median-conservative estimators, perform similar to (a) for $\lambda=s$ the graphical test from \citealp{bretz2009graphical}, i.e., the while-loop part from Algorithm~\ref{AlgorithmusMauerBretz} for $\lambda=s$, with the initial modified scaled levels based on the overall level $\alpha_M$ and the fixed stages $\tau_j^*$, $j\in\{1,\dots,m\}$. This gives the index set of rejected hypotheses $\bar{R}^s(\alpha_M)=\bar{R}_{\tau^*}^s(\alpha_M)$. Performing now Algorithm~\ref{alg_compromise_compatible} based on this index set and the scaled levels, one obtains the index set $\bar{R}^c(\alpha_M)$. Calculating the bounds from (\ref{mixtureCompatibleBounds}) with these index sets and the scaled levels, (common) median conservative estimators are obtained by $\hat{\vartheta}_{j,\tau^*}:=\bar{L}_{j,\tau^*}^{c}(\alpha_M)$, $j\in\{1,\dots,m\}$.
    \item[(c)] Assume that Algorithm~\ref{alg_repeated} has been executed for $\lambda\in\{r,s\}$ up to stage $k$, $1\leq k\leq K$. Perform now the Primary Algorithm~\ref{alg_1} with the (fixed) current stages $k_j^*(k)$ and replace in this algorithm $\alpha$ by $\alpha_M=0.5$. Then, (common) median conservative estimators are obtained by $\hat{\vartheta}_{j,k}:=L_{j,k}^{\lambda}(\alpha_M)$, $j\in\{1,\dots,m\}$.
    \item[(d)] Assume that Algorithm~\ref{alg_compromise}  for the overall level $\alpha$ has been executed (at stage $\tau^*$). For corresponding median-conservative estimators, perform now similar to (c) the Primary Algorithm~\ref{alg_1} for $\lambda=s$, with $\alpha_M$ and the fixed stages $\tau_j^*$, $j\in\{1,\dots,m\}$. This gives $L^s(\alpha_M)=L_{\tau^*}^{s}(\alpha_M)$. Perform now Algorithm~\ref{alg_compromise} with these bounds and the level $\alpha_M=0.5$. Then, (common) median conservative estimators are obtained by $\hat{\vartheta}_{j,\tau^*}:=L_{j,\tau^*}^{c}(\alpha_M)$, $j\in\{1,\dots,m\}$.
\end{itemize}
\end{lemma}
For (a) and (b) of Lemma~\ref{LemmaStrategiesEstimators} for all hypotheses rejected at the level $\alpha_M=0.5$, the estimators remain stuck at $0$ as long as not all hypotheses can be rejected. Because of increasing bounds with increasing evidence (property (b) in Definition~\ref{definitionInformativität}) this is in general not the case for the informative variants (c) and (d). They can serve as \emph{informative (common) median-conservative estimators}.

Additionally, please note that for (b) and (d) in Lemma~\ref{LemmaStrategiesEstimators}, as mentioned earlier, one can also calculate the SCI bounds and thus the (common) median conservative estimators at each stage $k$, $1\leq k\leq K$. 

If the estimators are calculated after each stage $k$, $1\leq k\leq K$, regardless of which variant (a) to (d) is used, the question of monotonicity of the estimators arises. For (a) and (c) for $\lambda=s$ it holds $\hat{\vartheta}_{j,k}\leq\hat{\vartheta}_{j,k+1}$, for all $1\leq k<K$ and $j\in\{1,\dots,m\}$ which follows directly from the properties of the SCI bounds. For (a) and (c), for $\lambda=r$, there is only the chance of monotonicity if the data collection for a hypothesis is stopped at the first stage where the hypothesis is rejected, i.e., the confidence bound fulfils $L_{j,k}^{r}(\alpha_M)\geq 0$ or $\bar{L}_{j,k}^{r}(\alpha_M)\geq 0$, $j\in\{1,\dots,m\}$. If one proceeds as described in Lemma~\ref{LemmaMonotonicityLambdaR}, i.e. one stops the data collection if the first time $L_{j,k}^{r}(\alpha)\geq 0$ is met, which implies monotonicity in those components, this does not automatically transfer to the estimators $\hat{\vartheta}_{j,k}:=L_{j,k}^r(\alpha_M)$. For the efficient multiple adjustment strategies (b) and (d), there is also no guaranteed monotonicity. 


\section{Summary and Discussion}\label{sectionSummaryAndDiscussion}
This paper consists of three main parts: the introduction and discussion of graphical group-sequential tests; the exploration of corresponding simultaneous confidence intervals, with a focus on informative bounds; and strategies for constructing common median-conservative estimators. Two of the investigated graphical group sequential tests are based on those from \citealp{maurer2013multiple}. One procedure works with repeated p-values and the other one with sequential p-values ($\lambda=r$ and $\lambda=s$ variant of Algorithm~\ref{AlgorithmusMauerBretz}). The sequential variant is more powerful because it allows for rejecting a hypothesis at a current stage via a local p-value obtained at an earlier stage, and all rejections that have already been made at earlier stages are retained. Thus, the sequential option consistently allows the use of evidence that comes only from data from earlier stages. In comparison, the $\lambda=r$ variant of Algorithm~\ref{AlgorithmusMauerBretz} is generally backward inconsistent. Rejections from earlier stages are kept despite possible insufficient evidence at subsequent stages, but a rejection at the current stage of a hypothesis that has not yet been rejected is only possible by using all the data up to the current stage. A natural desire is to have a test strategy that consistently considers for all hypotheses the evidence of all the data up to the current stage of data collection, and is powerful. This has motivated us to introduce the graphical group sequential procedure with efficient multiple adjustments given by Algorithm~\ref{alg_compromise_compatible}.

For the three introduced graphical group sequential designs, corresponding simultaneous confidence intervals have been investigated. This includes compatible bounds, which are in general not informative, and bounds that are not exactly compatible with the original underlying graphical group sequential test but are informative. All informative algorithms presented are based on a so-called \emph{primary algorithm} that is an extension of the algorithm from \citealp{informativeSCIBrannathKlugeScharpenberg}. We have added the calculation of a lower conservative approximation by the calculation of an anti-conservative upper approximation. This allows us to estimate the precision of the lower approximation and offers the possibility of a new stop criterion. One key property of the presented SCIs is that they can be calculated after each stage of the trial. We have investigated monotonicity properties across the stages. The compatible and informative SCI bounds corresponding to the sequential variant of Algorithm~\ref{AlgorithmusMauerBretz} are the only ones that are always increasing in all components across the stages. For the newly introduced efficient multiple adjustment strategies, this does not generally apply, as the local p-values at the current stage are always decisive for the size of the SCI bounds.

In Section~\ref{SectionMedianConservatEstimators}, we have discussed common median conservative estimators. The need of unbiased estimation in adaptive and group sequential trials is strongly emphasised in the \citealp{ICH-E20} guideline. The presented common median conservative estimators that were based on the informative SCI bounds have the decisive advantage that except for negligible cases the treatment effects of rejected hypotheses are not estimated at $0$. This does not apply to the estimators based on the compatible bounds. 

For the calculation of the informative SCI bounds and the common median conservative estimators, the \textsf{R}-package \emph{informativeSCI} will be extended.

The following extensions to the concepts and procedures introduced in this paper are conceivable: First, for the Primary Algorithm~\ref{alg_1}, for the different hypotheses, different information weights $q_j\in(0,1)$, $j=1,\dots,m$, could be used as mentioned in \citealp{informativeSCIBrannathKlugeScharpenberg}. This also applies to all strategies presented in this paper for constructing informative bounds for graphical group-sequential designs. Since group sequential designs are special cases of adaptive (group sequential) designs, it could be investigated how to extend the concept of informative SCI bounds and median conservative estimators of this paper to adaptive graph-based multiple testing procedures like those from \citealp{klinglmueller2014adaptive}. In addition, it could be investigated how the concept could be transferred to (one-stage, group sequential, and adaptive) monotone and non-monotone closure tests. For closed testing in two-stage adaptive designs \cite{magirr2013simultaneous} have suggested compatible simultaneous confidence intervals. A necessary condition for these bounds to be informative is that all hypotheses considered in the second stage are rejected. One could investigate how to extend the concept of informative bounds introduced in this paper to obtain SCI bounds that are also informative in the case that not all the hypotheses considered in the second stage are rejected.

\newpage

\section*{Acknowledgement}
The authors gratefully acknowledge the support of the Leibniz ScienceCampus Bremen Digital Public Health (www.digital-public-health.de), which is jointly funded by the Leibniz Association (W72/2022), the Federal State of Bremen, and the Leibniz Institute for Prevention Research and Epidemiology – BIPS.
\bibliographystyle{apalike}
\bibliography{references} 
\newpage
\appendix
\section{Numerical Implementation}
The developed R-package informativeSCI, which is available from CRAN or GitHub (\url{https://github.com/LianeKluge/informativeSCI}), will be extended to the calculation of the informative simultaneous confidence intervals for graphical group sequential test procedures.

For the numerical implementation, further properties of the repeated and sequential p-values were needed, which are proven in the following. 
\begin{lemma}
    Under the assumptions of the main article, we have for a component $1\leq j\leq m$ and a stage $1\leq k\leq K$:
    \begin{itemize}
        \item[(a)] The repeated p-values $p_{j,k}^r$ are continuous, strictly increasing functions (in $\mu_j$) that are bijective from $[-\infty, p_{j,k}^{-1}(\alpha_{j,k}^*(1))]$ to $[0,1]$. The inverse function is continuous and strictly increasing and given by $$(p_{j,k}^r)^{-1}(\gamma)=p_{j,k}^{-1}(\alpha_{j,k}^*(\gamma)),\quad\gamma\in[0,1];$$
        \item[(b)] The sequential p-values $p_{j,k}^s$ are continuous, strictly increasing functions (in $\mu_j$) that are bijective from $[-\infty, \max_{s=1,\dots,k}\{p_{j,s}^{-1}(\alpha_{j,s}^*(1))\}]$ to $[0,1]$. The inverse function is continuous and strictly increasing and given by $$(p_{j,k}^s)^{-1}(\gamma)=\max\left\{(p_{j,s}^r)^{-1}(\gamma):s=1,\dots,k\right\}=\max\left\{p_{j,s}^{-1}(\alpha_{j,s}^*(\gamma)):s=1,\dots,k\right\},\quad\gamma\in[0,1];$$
    \end{itemize}
\end{lemma}
\begin{proof}
    First, we prove part (a). Note that the nominal significance levels $\alpha_{j,k}^*(\gamma)$ are assumed to be continuous, strictly increasing, and bijective from $[0,1]$ to $[0,\alpha_{j,k}^*(1)]$ and thus invertible. In addition, for all $\mu_j$ such that $p_{j,k}(\mu_j)\leq\alpha_{j,k}^*(1)$ from the definition of the repeated p-value $p_{j,k}^r=\sup\{\gamma: p_{j,k}(\mu_j)>\alpha_{j,k}^*(\gamma)\}$ it follows that the repeated p-value equals the $\gamma$ solving $p_{j,k}(\mu_j)=\alpha_{j,k}^*(\gamma)$. Thus $p_{j,k}^r(\mu_j)=(\alpha_{j,k}^*)^{-1}(p_{j,k}(\mu_j))$. Because the local p-values are also invertible, we conclude $(p_{j,k}^r)^{-1}(\gamma)=p_{j,k}^{-1}(\alpha_{j,k}^*(\gamma))$, $\gamma\in[0,1]$ showing the statement. 

    We now prove part (b) by showing that, in general, it holds for functions $f_1,\dots,f_k$ with $f_s:[-\infty,c_s]\rightarrow[0,1]$ strictly increasing and existing strictly increasing inverse $f_s^{-1}:[0,1]\rightarrow[-\infty,c_s]$, $s=1,\dots,k$, that
    \begin{align}
    \min\{f_1,\dots,f_k\}^{-1}=\max\{f_1^{-1},\dots,f_k^{-1}\}.\label{statementInverse}
    \end{align}
    
    Define $f:=\min\{f_1,\dots , f_k\}$ and $g:=\max\{f_1^{-1},\dots,f_k^{-1}\}$. First, we argue that $g(f(\mu_j))=\mu_j$ for all $\mu_j\in[-\infty,\max_{s=1,\dots,k}c_s]$. It holds for all $s=1,\dots,k$ that $f(\mu_j)\leq f_s(\mu_j)$ and because the inverse functions $f_s^{-1}$ are strictly increasing it follows $f_s^{-1}(f(\mu_j))\leq f_s^{-1}(f_s(\mu_j))=\mu_j$ which implies $g(f(\mu_j))=\max\{f_1^{-1}(f(\mu_j)),\dots f_k^{-1}(f(\mu_j))\}\leq\mu_j$. For at least one index $s'$ with $1\leq s'\leq k$ it is $f(\mu_j)=f_{s'}(\mu_j)$ yielding $f_{s'}^{-1}(f(\mu_j))=\mu_j$ and thus we have $g(f(\mu_j))=\mu_j$. 
    
    Now we argue that $f(g(\gamma))=\gamma$ for all $\gamma\in[0,1]$. It holds for all $s=1,\dots,k$ that $g(\gamma)\geq f_s^{-1}(\gamma)$. Because the functions $f_s$ are strictly increasing we have $f_s(g(\gamma))\geq f_s(f_s^{-1}(\gamma))=\gamma$ yielding $f(g(\gamma))=\max\{f_1^{-1}(g(\gamma)),\dots,f_k^{-1}(g(\gamma))\}\geq\gamma$. Again, there is at least one index $1\leq s'\leq k$ with $g(\gamma)=f_{s'}^{-1}(\gamma)$ and thus $f_{s'}(g(\gamma))=\gamma$. In summary $f(g(\gamma))=\gamma$ and (\ref{statementInverse}) is true. Together with part (a), we conclude that statement (b) holds true.
\end{proof}

\section{Proof of upper estimate}
We argue in this section that the calculation of the approximation $\rho^{(s)}$ of $L_k^{\lambda}$, $\lambda\in\{r,s\}$ in Algorithm \ref{alg_1} provides indeed an (upper estimate) of the simultaneous confidence bounds. This holds not only true for the graphical group sequential tests from \citealp{maurer2013multiple} but also for all graphical test procedures and p-values satisfying the assumptions in \citealp{informativeSCIBrannathKlugeScharpenberg}.

\begin{lemma}
Under the assumptions of the Primary Algorithm \ref{alg_1}, for all valid starting vectors $\rho^{(0)}$ the sequence $(\rho^{(\ell)})_{\ell\in\mathbb{N}_0}$ is decreasing and converges from above against the SCI bounds $L_k^{\lambda}=(L_{k,1}^{\lambda},\dots, L_{k,m}^{\lambda})$.
\end{lemma}
\begin{proof}
First, we show inductively that the sequence $(\rho^{(\ell)})_{\ell\in\mathbb{N}_0}$ exists and is monotonically decreasing. Let $\ell=0$. The existence of a valid start vector $\rho^{(0)}$ is argued in the main article. Then, starting with equation (\ref{eq_astartUpper}) we have because of $\delta^{(1)}<\delta^{(0)}$ for each $j$ that 
\begin{align}
    p_{j,k_j^*}^{\lambda}(\rho_j^{(0)})\geq q^{(\rho_j^{(0)}\vee 0)}\nu_j(\rho^{(0)})(\alpha+\delta^{(1)})\label{equationproofUpperEntimate1}
\end{align}
holds true, $1\leq j\leq m$. Now, we define the two functions $g_{1,j}:\mathbb{R}_{-\infty}\rightarrow[0,1], \rho_j\mapsto p_{j,k_k^*}^{\lambda}(\rho_j)$, and $g_{2,j}:\mathbb{R}_{-\infty}\rightarrow\mathbb{R}_{\geq 0}, \rho_j\mapsto q^{(\rho_j\vee 0)}\nu_j(\rho^{(0)})(\alpha+\delta^{(1)})$. Here $\mathbb{R}_{-\infty}:=\mathbb{R}\cup\{-\infty\}$. The function $g_{1,j}$ is strictly monotonically increasing in $\rho_j$, and the function $g_{2,j}$ is monotonically decreasing in $\rho_j$. Using these properties and (\ref{equationproofUpperEntimate1}) we have that for an $\rho_j^{(1)}$ such that $g_{1,j}(\rho_j^{(1)})=g_{2,j}(\rho_j^{(1)})$ it must hold $\rho_j^{(1)}\leq \rho_j^{(0)}$. Finding $\rho_j^{(1)}$ such that $g_{1,j}(\rho_j^{(1)})=g_{2,j}(\rho_j^{(1)})$, $1\leq j\leq m$, are the key equations that are to be solved in Algorithm~\ref{alg_1}. Because $g_{1,j}$ tends to $0$ if $\rho_j$ approaches minus infinity and is equal to $0$ if $\rho_j=-\infty$ and $g_{2,j}$ is a non-negative function, we have the existence of a unique $\rho_j^{(1)}\in\mathbb{R}_{-\infty}$ that meets the equation and it holds $\rho_j^{(1)}\leq \rho_j^{(0)}$, $1\leq j\leq m$.

Because we have $\rho^{(1)}\leq\rho^{(0)}$ pointwise and $\nu_j$ is monotonically increasing in all components we obtain for all $1\leq j\leq m$ that 
\begin{align}
    p_{j,k_j^*}^{\lambda}(\rho_j^{(1)})=g_{1,j}(\rho_j^{(1)})= g_{2,j}(\rho_j^{(1)})= q^{(\rho_j^{(1)}\vee 0)}\nu_j(\rho^{(0)})(\alpha+\delta^{(1)})\geq q^{(\rho_j^{(1)}\vee 0)}\nu_j(\rho^{(1)})(\alpha+\delta^{(1)})~.\label{StartingEquationUpper}
\end{align}
Now, we can proceed for $\rho^{(2)}$ like for $\rho^{(1)}$ and we obtain a component-wise decreasing sequence $(\rho^{(\ell)})_{\ell\in\mathbb{N}_0}$. 

As a next step, we show that all elements of the sequence are greater than the projection $L_k^{\lambda}$ (component-wise).  First, equation (\ref{StartingEquationUpper}) holds similarly for all elements of the sequence. Let $\ell\in\mathbb{N}_0$. For all $j\in\{1,\dots,m\}$ such that $\nu_j(\rho^{(\ell)})>0$ it holds 
\begin{align}
    p_{j,k_j^*}^{\lambda}(\rho_j^{(\ell)})\geq q^{(\rho_j^{(\ell)}\vee 0)}\nu_j(\rho^{(\ell)})(\alpha+\delta^{(\ell)})>q^{(\rho_j^{(\ell)}\vee 0)}\nu_j(\rho^{(\ell)})\alpha=\alpha_j^{\rho^{(\ell)}}~.\label{strictInequality}
\end{align}
To obtain the strict inequality in (\ref{strictInequality}), it is important that $\delta^{(\ell)}>0$. If $\nu_j(\rho^{(\ell)})=0$, for the local significance level it obviously holds $\alpha_j^{\rho^{(\ell)}}=q^{(\rho_j^{(\ell)}\vee 0)}\nu_j(\rho^{(\ell)})\alpha=0$. We use now the definition of the confidence set $C_k^{\lambda}$ from (\ref{confidenceSetEachStage}) and conclude that for all $\ell\in\mathbb{N}_0$ we have $\rho^{(\ell)}\in C_k^{\lambda}$. This implies $\rho^{(\ell)}\geq L_k^{\lambda}$ component-wise.  

It remains to show the convergence of the sequence $(\rho^{(\ell)})_{\ell\in\mathbb{N}_0}$ against $L_k^{\lambda}$. For all components $j\in\{1,\dots,m\}$ for whom $L_{k,j}^{\lambda}>-\infty$, the sequence $(\rho_j^{(\ell)})_{\ell\in\mathbb{N}_0}$ converges because it is monotonically decreasing and bounded from below. For the other components with $L_{k,j}^{\lambda}=-\infty$, there are initially two possibilities. Either the sequence is bounded from above by another $c_j>-\infty$ and thus also converges, or it is not bounded and thus diverges against $-\infty=L_{k,j}^{\lambda}$. We show that the latter is always the case. Let $\rho^{\infty}$ be the limit vector of $(\rho^{(\ell)})_{\ell\in\mathbb{N}_0}$ (i.e. the components converge against $\rho^{\infty}_j$ if $\rho^{\infty}_j>-\infty$ or diverge against $\rho^{\infty}_j=-\infty$, $1\leq j\leq m$). It is important that the functions $\mu_j\mapsto q^{\mu_j\vee 0}$ and $\mu\mapsto\nu_j(\mu)$ as well as $\mu\mapsto\alpha_j^{\mu}$, $j\in\{1,\dots,m\}$, can be (uniquely) continuously extended to $\mu\in\mathbb{R}^m_{\infty}$. By using the continuity, the convergence of $(\delta^{(\ell)})_{\ell\in\mathbb{N}_0}$ against $0$ and, in particular, equation (\ref{eq_stepUpper}) from Algorithm \ref{alg_1}, we conclude
\begin{align}
    p_{j,k_j^*}^{\lambda}(\rho_j^{\infty})=q^{(\rho_j^{\infty}\vee 0)}\nu_j(\rho^{\infty})\alpha=\alpha_j^{\rho^{\infty}},\quad\text{for all~}1\leq j\leq m~.\label{equation_unique}
\end{align}
From \citealp{informativeSCIBrannathKlugeScharpenberg} we now that (\ref{equation_unique}) uniquely characterizes $L_k^{\lambda}=(L_{k,1}^{\lambda},\dots L_{k,m}^{\lambda})$. In summary, $\rho^{\infty}=L_k^{\lambda}$ component-wise and the sequence $(\rho^{(\ell)})_{\ell\in\mathbb{N}_0}$ converges from above against $L_k^{\lambda}$ for each valid starting vector $\rho^{(0)}$.

\end{proof}

\section{Further proofs}
\subsection{Properties of the informative bounds for $\lambda=r$}
\begin{lemma}
Let $\lambda=r$ in Algorithm \ref{alg_repeated} and assume for all  $j\in\{1,\dots,m\}$ with $j\in\cup_{k=1}^{\tau^*}R_k^r$ that it holds $\tau_j^*\leq r_j:=\min\{1\leq k\leq\tau^*~|~j\in R_{k}^r\}$, i.e. if a hypothesis is rejected the first time at a certain stage the data collection for this hypothesis is stopped or has already been stopped at an earlier stage. Then, the index sets of rejected hypotheses by Algorithm \ref{alg_repeated} fulfil $R_{k}^r\subseteq R_{k+1}^r$, $1\leq k<\tau^*$. In particular, for all $j\in R_{\tau^*}^{r}$ we have for all $k$ with $r_j\leq k <\tau^*$ that $L_{j,k}^{r}\leq L_{j,k+1}^{r}$, i.e. the confidence bounds increase in these components.
\end{lemma}
\begin{proof}
We show the result inductively. Let $1\leq k<\tau^*$ be the smallest stage such that a hypothesis is rejected, i.e., $R_k^{r}\neq \varnothing$. From \citealp[Appendix B.1]{informativeSCIBrannathKlugeScharpenberg}, we know that
\begin{align}
    p_{j,k_j^*(k)}^r(L_{j,k}^r)=q^{(L_{j,k}^r\vee 0)}\nu_j(L_k^r)\alpha =\alpha_j^{L_k^r}, \quad 1\leq j\leq m,\label{limitEquation}
\end{align}
holds true. We assume that for all $j\in R_k^{r}$ the data collection is stopped at stage $k$ or was already stopped at an earlier stage, i.e. $\tau_j^*\leq k$. Let the trial continue to the next stage $k+1$ as described in Algorithm~\ref{alg_repeated}. We show now that
\begin{align}
    \mu_j^{(0)}=\begin{cases}
             \min\{0, (p_{j,k_j^*(k+1)}^{r})^{-1}(\alpha_j)\}~, & j\notin R_{k}^r\\
            L_{j,k}^r~, & j\in R_{k}^r
		 \end{cases}~,\quad 1\leq j\leq m,\label{startingVectorFromProof}
\end{align}
is a valid starting vector for the lower approximation in Algorithm~\ref{alg_1}) for calculating the bounds $L_{k+1}^r=(L_{1,k+1}^r,\dots, L_{m,k+1}^r)$. We need the following observation: From the construction of the local levels $\alpha_j^{\mu}$, $1\leq j\leq m$, by the dual graphs as described in Section~\ref{sectionCharacteristics} or \citealp{informativeSCIBrannathKlugeScharpenberg}, we can easily see that $\alpha_j^{\mu}$ does not depend on $\mu_i$ as long as $\mu_i\leq 0$, $1\leq i\leq m$. This follows from the fact that in the dual graph there are no arrows starting from $H_i^{\mu_i}$ and going to other hypotheses as long as $\mu_i\leq 0$. Using this, equation (\ref{limitEquation}), and that for all $j\in R_k^r$ it holds $\tau_j^*\leq k$ which implies $k_j^*(k)=k_j^*(k+1)$ and thus $p_{j,k_j^*(k+1)}^r= p_{j,k_j^*(k)}^r$, we can conclude for each $j\in R_k^r$ that
\begin{align*}
    p_{j,k_j^{*}(k+1)}^r(\mu_j^{(0)})=p_{j,k_j^{*}(k+1)}^r(L_{j,k}^r)=q^{(\mu_j^{(0)}\vee 0)}\nu_j(\mu^{(0)})\alpha~.
\end{align*}
Thus, the starting-vector condition for the lower approximation is fulfilled for the components $j\in R_k^r$. Let now $j\notin R_k^r$. As mentioned in Section~\ref{sectionCharacteristics} we have that $\alpha_j^{\mu}\geq \alpha_j$ if $\mu_j\leq 0$. The reason for this is that for $\mu_j\leq 0$ in the dual graph, no level is transferred from $H_j^{\mu_j}$ to other hypotheses, implying that the starting level $\alpha_j$ remains at $H_j^{\mu_j}$. Using this, we can conclude for all $j\notin R_k^r$ that $\alpha_j^{\mu^{(0)}}\geq \alpha_j$ and the definition from (\ref{startingVectorFromProof}) gives $p_{j,k_j^*(k+1)}^r(\mu_j^{(0)})\leq \alpha_j^{\mu^{(0)}}$.

In summary, we have that $\mu^{(0)}$ as defined in (\ref{startingVectorFromProof}) is a valid lower starting vector for calculating the bounds $L_{k+1}^r=(L_{1,k+1}^r,\dots, L_{m,k+1}^r)$ with Algorithm~\ref{alg_1}. This implies using Theorem~\ref{theoremMainResults} that $\mu^{(0)}\leq L_{k+1}^r$ component-wise. In particular we have for all $j\in R_k^{r}$ that $L_{j,k+1}^r\geq L_{j,k}^r\geq 0$ which gives $j\in R_{k+1}^r=\{1\leq j\leq m: L_{j,k+1}^r\geq 0\}$, i.e. $R_k^r\subseteq R_{k+1}^r$. If we update now $k\rightarrow k+1$, we can prove the statement inductively with the same arguments as before.
\end{proof}
\subsection{Properties of the informative bounds for $\lambda=s$}
\begin{lemma}
Let $\lambda=s$ in Algorithm~\ref{alg_repeated} and $2\leq k\leq\tau^*$. A valid starting vector for the lower approximation of $L_k^s$ is then given by $\mu^{(0)}:=L_{k-1}^s$.
\end{lemma}
\begin{proof}
Let $2\leq k\leq\tau^*$. From \citealp[Appendix B.1]{informativeSCIBrannathKlugeScharpenberg}, we know for the confidence bounds at stage $k-1$ that
\begin{align*}
    p_{j,k_j^*(k-1)}^s(L_{j,k-1}^s)=q^{(L_{j,k-1}^s\vee 0)}\nu_j(L_{k-1}^s)\alpha, \quad 1\leq j\leq m,
\end{align*}
holds true. If we use the monotonicity property of the sequential p-values across the stages and $k_j^*(k-1)\leq k_j^*(k)$, $1\leq j\leq m$, we obtain
\begin{align*}
     p_{j,k_j^*(k)}^s(L_{j,k-1}^s)\leq q^{(L_{j,k-1}^s\vee 0)}\nu_j(L_{k-1}^s)\alpha , \quad 1\leq j\leq m.
\end{align*}
Thus, the condition for the lower starting vector in Algorithm~\ref{alg_1} is fulfilled. 
\end{proof}
\subsection{Properties of the informative bounds for GSD with efficient multiple adjustment}
\begin{lemma}
    For Algorithm~\ref{alg_compromise} there exist unique $L_{j,\tau^*}^c\in\mathbb{R}_{-\infty}$, $L_{j,\tau^*}^c\leq L_{j,\tau^*}^s$, fulfilling (\ref{keyEquationMixtureInformative}), $j\in\{1,\dots,m\}$, and a valid search area for the bisection search performed is given by (\ref{startVectorMixtureInformative}).
\end{lemma}
\begin{proof}
From \citealp[Appendix~B.1]{informativeSCIBrannathKlugeScharpenberg} we know that
\begin{align*}
    p_{j,\tau^*_j}^s(L_{j,\tau^*}^s)=q^{(L_{j,\tau^*}^s\vee 0)}\nu_j(L_{\tau^*}^s)\alpha=\alpha_j^{L_{\tau^*}^s},\quad 1\leq j\leq m~.
\end{align*}
By using $p_{j,\tau_j^*}^r\geq p_{j,\tau_j^*}^s$ and the definition of $\tilde{\omega}_j$ from Algorithm~\ref{alg_compromise} we have
\begin{align*}
      p_{j,\tau^*_j}^r(L_{j,\tau^*}^s)\geq q^{(L_{j,\tau^*}^s\vee 0)}\nu_j(L_{\tau^*}^s)\alpha=\tilde{\omega}_j(L_{j,\tau^*}^s)\alpha,\quad 1\leq j\leq m~.
\end{align*}
Now we need the following observation. From Proposition~\ref{propositionWithImportantProperties} (ii) or \citealp[Proposition~4,]{informativeSCIBrannathKlugeScharpenberg} it follows that for each $j\in\{1,\dots,m\}$ it holds that $\alpha_j^{\mu}=q^{\mu_j\vee 0}\nu_j(\mu)\alpha$ is increasing in all components $i\neq j$ because this the case for $\nu_j$ and $q^{\mu_j\vee 0}$ does not depend on the $i$-th component. Due to $\sum_{\ell=1}^m\alpha_{\ell}^{\mu}=\alpha$ it holds that $\alpha_j^{\mu}$ is decreasing in $\mu_j$. Thus, we obtain that $\tilde{\omega}_j(\mu_j)$ is a decreasing function of $\mu_j$. In addition for $g_{1,j}:\mathbb{R}_{-\infty}\rightarrow[0,1]$, $\mu_j\mapsto p_{j,\tau_j^*}^r(\mu_j)$, we have that $g_{1,j}$ approaches $0$ if $\mu_j$ approaches $-\infty$. Using all these arguments we find a unique $L_{j,\tau^*}^c\in\mathbb{R}_{-\infty}$ such that $$ p_{j,\tau^*_j}^r(L_{j,\tau^*}^c)= \tilde\omega_j(L_{j,\tau^*}^c)\alpha$$ is met. On $\mu_j\in [-\infty, 0]$ the level $\tilde{\omega}_j(0)\alpha$ is constant implying that in the case of $\tilde{\omega}_j(0)\alpha>0$ the unique solution $L_{j,\tau^*}^c$ is greater than $-\infty$ and lies in 
\begin{align*}
L_{j,\tau^*}^c\in\left[(p_{j,\tau_j^*}^r)^{-1}(\tilde\omega_j(0)\alpha)\,, L_{j,\tau^*}^s \right]~.
\end{align*}
In the other case, the unique solution is given by $L_{j,\tau^*}^c=-\infty$ and no bisection search is necessary. 
\end{proof}

\begin{proposition}\label{PropositionInformativenessMixture}
The simultaneous confidence bounds $L_{\tau^*}^c$ obtained by Algorithm~\ref{alg_compromise} satisfy Definition~\ref{definitionInformativität} with the understanding of evidence from Definition~\ref{DefinitionEvidenzMixture}.
\end{proposition}

\begin{proof}
In the following we use the notations $L_{j,\tau^*}^c(x)$ and $L_{j,\tau^*}^s(x)$, $j\in\{1,\dots,m\}$, to indicate the data dependency of the SCI bounds. We also use the notation  $\tilde{\omega}_j(\mu_j,x)$ to indicate the data dependency of the weights on $x$ because of the dependency on $L_{i,\tau^*}^s(x)$, $i\in\{1,\dots,m\}$, $i\neq j$.

First, part (a) of Definition~\ref{definitionInformativität} is shown. If a hypothesis $H_j$, $j\in\{1,\dots,m\}$ has no gatekeeper it holds $\alpha_j>0$ implying that $\tilde{\omega}_j(0)\alpha\geq\alpha_j>0$ and thus $L_{j,\tau^*}^c>-\infty$ due to Remark~\ref{remarkRealValuedSearchAreaBisectionMixture}. If for a gatekeeper $H_i$ for a hypothesis $H_j$ (in the original graph), $i,j\in\{1,\dots,m\}$, $i\neq j$, it holds $L_{i,\tau^*}^c>0$ this implies $L_{i,\tau^*}^s>0$ because of $L_{i,\tau^*}^s\geq L_{i,\tau^*}^c$. This means that some level is transferred from the $i$-th component to the $j$-th component, yielding in particular that $\tilde{\omega}_j(0)\alpha>0$. This implies $L_{j,\tau^*}^c>-\infty$, again by Remark~\ref{remarkRealValuedSearchAreaBisectionMixture}.

To prove part (b) of Definition~\ref{definitionInformativität}, let $x,x^{\prime}$ be two data sets and $j\in\{1,\dots,m\}$ a component such that conditions (i) and (iii) of Definition~\ref{definitionInformativität} are met. For all other components $i\in\{1,\dots,m\}$, $i\neq j$, let part (ii) of Definition~\ref{definitionInformativität} be met. Because of (i) and (ii), we find, looking at Definition~\ref{DefinitionEvidenzMixture}, that for all components $\ell\in\{1,\dots,m\}$ it holds $p_{\ell,\tau_{\ell}^*(x^{\prime})}^s(x^{\prime},\mu_{\ell})\leq p_{\ell,\tau_{\ell}^*(x)}^s(x,\mu_{\ell})$ for all $\mu_{\ell}\in\mathbb{R}$. This gives $C_{\tau^*}^{s}(x^{\prime})\subseteq C_{\tau^*}^{s}(x)$ for the confidence set from (\ref{confidenceSetRepeatlyCallingPrimary}) implying that for the projections it holds $L_{\tau^*}^s(x)\leq L_{\tau^*}^s(x^{\prime})$ component-wise. Now we have that $\tilde{\omega}_j(\mu_j,x)\leq \tilde{\omega}_j(\mu_j,x^{\prime})$ for all $\mu_j\in\mathbb{R}.$ Using (iii) of Definition~\ref{definitionInformativität}, i.e. $L_{j,\tau^*}^c(x)>-\infty$, we obtain $\tilde{\omega}_j(0,x)\alpha>0$. This yields that the solution $L_{j,\tau^*}^c(x)$ of $p_{j,\tau^*_j(x)}^r(L_{j,\tau^*}^c(x))=\tilde{\omega}(L_{j,\tau^*}^c(x),x)\alpha$ and the solution $L_{j,\tau^*}^c(x^{\prime})$ of $p_{j,\tau^*_j(x^{\prime})}^r(L_{j,\tau^*}^c(x^{\prime}))=\tilde{\omega}(L_{j,\tau^*}^c(x^{\prime}),x^{\prime})\alpha$ are both greater than $-\infty$. By knowing this and using part (i) of Definition~\ref{definitionInformativität}, i.e. $p_{j,\tau_j^{*}(x^{\prime})}(x^{\prime},\mu_j)<p_{j,\tau_j^{*}(x)}(x,\mu_j)$ for all $\mu_j\in\mathbb{R}$, we obtain that $L_{j,\tau^*}^c(x^{\prime})>L_{j,\tau^*}^c(x)$. This proves part (b) of Definition~\ref{definitionInformativität}.
\end{proof}

\begin{remark}
    One can see from the proof above that the bounds $L_{\tau^*}^c$ obtained by Algorithm~\ref{alg_compromise} also satisfy a stricter informativeness definition. For instance, it can easily be seen that part (a) of Definition~\ref{definitionInformativität} is also fulfilled when for a hypothesis $H_j$ with at least one gatekeeper $H_i$ it holds $L_{i,\tau^*}^s>0$ instead of $L_{i,\tau^*}^c>0$, $j,i\in\{1,\dots,m\}$, $i\neq j$. This follows from the fact that if for a gatekeeper it holds $L_{i,\tau^*}^s>0$ this implies $\tilde{\omega}_j(0)\alpha>0$. Then, due to Remark~\ref{remarkRealValuedSearchAreaBisectionMixture} we have $L_{j,\tau^*}^c>-\infty$.

    In addition, $L_{j,\tau^*}^c(x^{\prime})>L_{j,\tau^*}^c(x)$, $j\in\{1,\dots,m\}$ in part (b) of Definition~\ref{definitionInformativität} is also met if instead of (ii) a weaker condition is fulfilled. With regard to the understanding of evidence from Definition~\ref{DefinitionEvidenzMixture} only part (a) and (c) must be fulfilled for all components $i\neq j$ because they shape $\tilde{\omega}_j(\mu_j)\alpha$, $\mu_j\in\mathbb{R}$.
\end{remark}

\end{document}